\documentclass[11pt,a4paper]{article}

\pdfoutput=1 

\newcommand{\ifConferenceVersion}{\iffalse}
\newcommand{\ifJournalVersion}{\iftrue}

\makeatletter
\@ifundefined{ifConferenceVersion}{
    \newcommand{\ifConferenceVersion}{\iftrue}
    \newcommand{\ifJournalVersion}{\iffalse}
}{}
\makeatother

\ifConferenceVersion
    
    \newcommand{\InJournal}[1]{}
    \newcommand{\InConference}[1]{#1}
\else
    
    \newcommand{\InJournal}[1]{#1}
    \newcommand{\InConference}[1]{}
    \usepackage{amsthm}
    \usepackage{fullpage}
\fi

\usepackage{paralist}

\usepackage{amsmath}
\usepackage{amssymb}
\usepackage{color}
\usepackage{url}
\usepackage{graphicx}
\usepackage{ifpdf}

\newtheorem{theorem}{Theorem}[section]
\newtheorem{lemma}[theorem]{Lemma}

\newtheorem{claim}[theorem]{Claim}
\newtheorem{proposition}[theorem]{Proposition}
	
\newtheorem{definition}{Definition}[section]

\newcommand{\R}{\mathbf R}
\newcommand{\E}{\mathbf E}
\newcommand{\Var}{\mathrm {Var}}
\newcommand{\poly}{\mathrm {poly}}
\newcommand{\eps}{\varepsilon}
\newcommand{\const}{\mathrm {const}}
\newcommand{\modd}{\,\mathrm{mod}\,}
\newcommand{\cont}{\bar}
\newcommand{\contz}{\tilde}
\newcommand{\lnc}{r^{+}}
\newcommand{\lnf}{r^{-}}
\newcommand{\bigo}{\mathcal{O}}
\newcommand{\Oh}{\bigo}
\newcommand{\Otilde}{\widetilde \bigo}
\newcommand{\Omegatilde}{\widetilde \Omega}
\newcommand{\Thetatilde}{\widetilde \Theta}

\title{
{
On Convergence and Threshold Properties of\\ Discrete Lotka-Volterra 
 Population Protocols}
}

\author{
Jurek Czyzowicz~\footnote{Universit\'{e} du Qu\'{e}bec en Outaouais, Dep. d'Informatique, Gatineau, QC, Canada.}
\and Leszek Gasieniec~\footnote{University of Liverpool, Department of Computer Science, Liverpool, UK.}
\and Adrian Kosowski~\footnote{Inria Paris and LIAFA, Universit\'e Paris Diderot, France.}
\and Evangelos Kranakis~\footnote{Carleton University, School of Computer Science, Ottawa, ON, Canada.}
\and Paul~G.~Spirakis\footnotemark[2]\ \footnote{CTI, Patras, Greece.}
\and Przemys\l{}aw Uzna\'nski~\footnote{Helsinki Institute for Information Technology HIIT, Aalto University, Finland.}
}

\begin{document}
\maketitle

\begin{abstract}
In this work we focus on a natural class of population protocols whose dynamics are modelled by the discrete version of Lotka-Volterra equations. In such protocols, when an agent $a$ of type (species) $i$ interacts with an agent $b$ of type (species) $j$  with $a$ as the initiator, then $b$'s type becomes $i$ with probability $P_{ij}$. In such an interaction, we think of $a$ as the predator, $b$ as the prey, and the type of the prey is either converted to that of the predator or stays as is. Such protocols capture the dynamics of some opinion spreading models and
generalize the well-known Rock-Paper-Scissors discrete dynamics. We consider the pairwise interactions among agents that are scheduled uniformly at random.

We start by considering the convergence time and show that any Lotka-Volterra-type protocol on an $n$-agent population converges to some absorbing state in time polynomial in $n$, w.h.p., when any pair of agents is allowed to interact. By contrast, when the interaction graph is a star, even the Rock-Paper-Scissors protocol requires exponential time to converge. We then study threshold effects exhibited by Lotka-Volterra-type protocols with 3 and more species under interactions between any pair of agents. We start by presenting a simple 4-type protocol in which the probability difference of reaching the two possible absorbing states is strongly amplified by the ratio of the initial populations of the two other types, which are transient, but ``control'' convergence. We then prove that the Rock-Paper-Scissors protocol reaches each of its three possible absorbing states with almost equal probability, starting from any configuration satisfying some sub-linear lower bound on the initial size of each species. That is, Rock-Paper-Scissors is a realization of a ``coin-flip consensus'' in a distributed system. Some of our techniques may be of independent value.
\end{abstract}

\InConference{\small\noindent\textbf{\looseness-1 This paper is submitted to track A. Due to lack of space, the full paper with all proofs is~included in a clearly marked Appendix to be read at the Program Committee's discretion.}}

\section{Introduction}

Population protocols are a recent model of computation that captures the way in which the complex behavior of systems (biological, sensor nets, etc.) emerges from the underlying local interactions of agents. Agents are modeled as anonymous automata with a finite number of states, and interactions (changes of state) occur between randomly chosen pairs of agents under some fixed set of local rules. The interaction follows from the mobility of agents in the population, as in the case of birds flying past each other in a flock in the setting originally described by Angluin et al.~\cite{AADFP06,AAER07}. More generally, we can model agents as nodes of an interaction graph $G$, and assume interactions take place along the edges of this graph.

Population protocols provide a way of describing dynamical effects which may occur in a population. For example, one can imagine that members of a population can be either healthy or infected, and whenever two individuals meet, if one is infected, then the other one also becomes infected. Thus the interesting question becomes: how fast can the infection spread? Quite naturally, population protocols are also used to model opinion spread in populations under interactions. An interaction between a pair of agents, one holding opinion A and the other opinion B, results in a possible change of opinion by one of the interacting agents. Such local interactions result in a change over time of the relative sizes of the options holding opinions A and B. Eventually, the population protocol may lead the system to converge to a state in which one type, A or B, becomes dominant in the population. The probability of convergence to a given dominant type may potentially depend on the initial state of the population in different ways, e.g., exhibiting linear behavior, or transitions at one or more thresholds (cf. Fig.~\ref{fig0}). In fact, the ability of population protocols to converge to an opinion presented by at least some threshold ratio of the population (e.g., a majority) lies at the heart of their study (cf.~\cite{AR09,MCS11}).

In this work, we focus on a natural scenario of interactions modeled by the discrete version of Lotka-Volterra equations, with the goal of better understanding their applicability in the computational framework of opinion spreading and voting protocols. In their original form, the (continuous) Lotka-Volterra differential equation were initially applied in the modeling of periodic chemical reactions and also
in the predator-prey dynamics of fish in the Adriatic Sea~\cite[p.11]{HS98}, and are perhaps best known for their connection to replicator dynamics and to evolving strategies in game theory~\cite{HS98,SMJSRP14}. In discrete Lotka-Volterra-type (LV-type) population protocols, during an interaction, the initiating agent (holding some state A) tries to impose its state on the other agent (holding some other state B) and succeeds with some probability $P_{AB}$. The LV-type protocol on $k$ states is fully characterized by its $k\times k$ probability matrix $P$. LV-type interactions are natural both in the context of predator-prey protocols, in that they correspond to a possible expansion of the predating (initiating) agent into the ecological niche of its prey, and in opinion propagation, in that they do not allow a new derived state C to be created as a result of an interaction. Unlike their continuous variant, discrete LV-type protocols always converge to an absorbing state in which no further state changes occur. We study the time until such convergence happens in general, and look at the probability of achieving different absorbing states depending on the initial distribution of states in the population for specific protocols.

\begin{figure}[t]
\centering
\ifpdf
\includegraphics[width=0.9\textwidth,trim=0cm 6cm 0cm 6.5cm,clip]{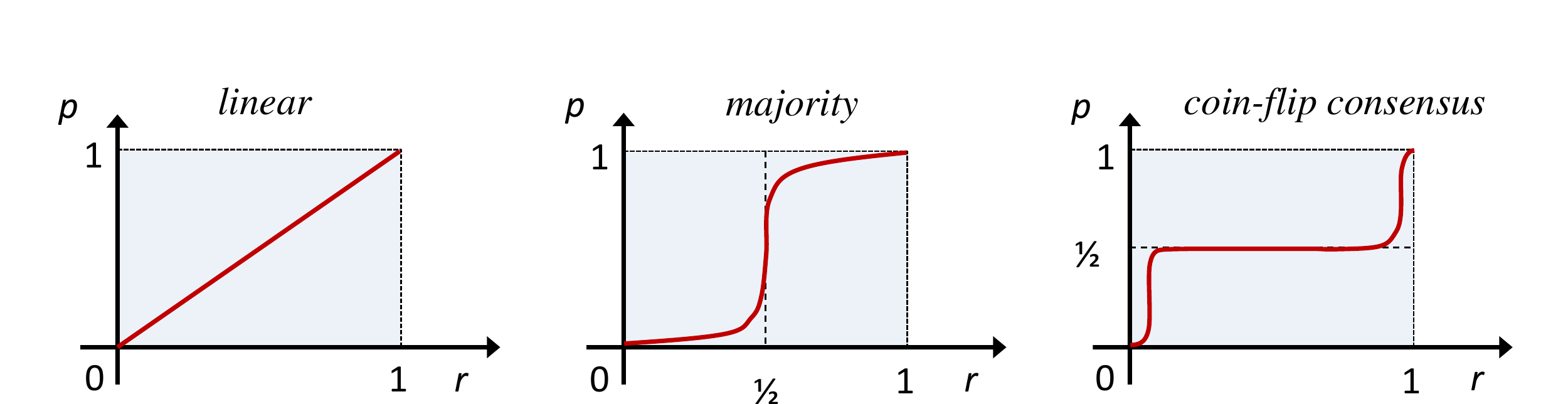}
\fi
\caption{Examples of objective functions for opinion spreading: probability $p$ that a given type becomes dominant in the population as a function of the fraction $r$ of its supporters in the initial population.\vspace{-3mm}}\label{fig0}
\end{figure}

\subsection{Our Results}

We start by proving in Section~\ref{sec:convergence} a general convergence result: any LV-type protocol on a $n$-agent population converges to some absorbing state in time $\Oh(poly(n))$, w.h.p., under the model of uniformly random interactions between agents (i.e., when the interaction graph is the complete graph $G=K_n$). Whereas studies of the behavior of the continuous LV dynamics under artificial stochastic noise of the random walk type have been performed for some particular dynamics in the statistical physics community (cf.~\cite{SMJSRP14} for a survey), this is, to the best of our knowledge, the first mathematically rigorous study of the impact of finite populations on the behavior of the dynamics, and the first one to present results which hold for an \emph{arbitrary} probability matrix $P$ of the protocol. Our proof takes advantage of the structure of the replicator dynamics corresponding to the continuous Lotka-Volterra equations, which admit either stationary orbits in the state space or have a repeller, depending on the parameters of the protocol.

By contrast, we also show in Section~\ref{sec:convergence} that introducing an interaction constraint can severely impact the convergence time for LV-type protocol. We consider a specific LV-type protocol known as \emph{rock-paper-scissors} (RPS), in which each of the three types overcomes exactly one other type in cyclic manner, and show that RPS requires exponential time to converge to an absorbing state when the interaction graph is a star ($G = K_{1,n}$).

Next, we look at the applicability of LV-type protocols in the context of their threshold behavior in voting problems which require a consensus of opinion. For the case of 2 types, the only unbiased LV-type protocol encompasses the so-called ``game of life and death'' between the 2 types, converging to a given absorbing state with probability proportional to its initial representation in the population (regardless of the interaction graph $G$). We show, however, that for 3 and more types, threshold effects become apparent even under uniform interactions ($G = K_n$). We start by proposing in Section~\ref{sec:majority} a simple 4-type majority-type protocol, in which the probability difference of reaching the two absorbing types is amplified with respect to the ratio of the initial populations of the two other states. We close the paper by exhibiting in Section~\ref{sec:rps}, for the before-mentioned RPS protocol a completely different type of threshold effect in the small population region. We prove that RPS reaches each of its three absorbing types with almost equal probability ($1/3 \pm o (1)$), starting from any configuration satisfying some sub-linear lower bound on the sizes of the three types. Our proof proceeds by a Martingale-type analysis, which may be of independent interest, and takes into account the symmetries of the state space of the protocol. We can thus view the RPS protocol as an embodiment of the ``coin-flip consensus'' illustrated in Fig.~\ref{fig0}(c): any opinion with non-negligible representation in the population, even a minority one, has an equal chance of success in the opinion-spreading process. To the best of our knowledge, this is the first population protocol with polynomial-time convergence for which such a property has been identified.


\subsection{Related Work}

\paragraph{Population Protocols and Majority Computation.}

The population protocol model of Angluin et al.\ \cite{AADFP06,AAE08} captures random interactions between finite-state agents, motivated by applications in sensor mobility. Despite the limited computational capabilities of individual sensors, such protocols permit at least (depending on available extensions to the model) the computation of two important classes of functions: threshold predicates, which decide if the weighted average of types appearing in the population exceeds a certain value, and modulo remainders of similar weighted averages. With respect to threshold behavior, a major problem is the design of majority protocols, which consist in obtaining convergence of all agents of the population to a type initially represented by the majority of agents. Such a protocol, converging to a population of a single type, was first proposed in~\cite{AAE08}. Given the complete interaction graph, the type reached by the protocol is the initial majority type, w.h.p., provided that the initial difference between the majority and minority type is $\omega(\sqrt n \log n)$ for a $n$-agent population. This protocol relies on $3$ types, two of which represent the original types present in the population, while the third is a transient type representing a blank opinion. A $4$-state protocol for finding a majority is presented in~\cite{AR09}, based on a different principle of ``leader'' and ``follower'' agents, and achieves similar performance guarantees. By contrast, \cite{MNRS14} presented the first protocol which converges to the initial majority type with probability $1$, even when the difference between types in the original population is constant. This protocol makes use of $4$ states and finds a majority in expected polynomial time, even in the case when interactions are not spread uniformly over the population, but restricted to a connected subgraph of agent pairs. (We remark that none of the mentioned majority protocols belongs to the LV-type considered in this paper, in particular, due to the creation of transient states which do not exist in the initial population.) Other applications and models of population protocols are surveyed in~\cite{AR09,K12,MCS11}.

\paragraph{Spreading of opinion and voting.}

The spread of trust and opinion in a social network was one of the original motivations for the study of population protocols~\cite{DF01}. Problems in which a set of nodes has to converge to a consensus decision chosen from a candidate set of values proposed by the participating nodes, are also of fundamental importance in distributed computing, in tasks such as serialization of database operations or leader election~\cite{HP01}. Models of voting processes, which solve such questions, involve the propagation of opinion through multiple push- or pull-operations between pairs of agents, usually performed in parallel throughout the system. From the perspective of security and simplicity of design, a desirable property of the protocols is that at any time during the execution, the state of the node should describe its current opinion, belonging to the set of opinions initially represented in the population. Under this constraint, given a set of only $2$ initial opinions, it is impossible to obtain convergence to the majority opinion w.h.p.\ of correctness in the standard model of voting~(cf.~e.g.~\cite{CER14}). However, majority voting can be achieved in many graph classes by extensions of the population protocol framework, allowing simultaneous interactions between more than $2$ nodes. Specifically, protocols in which a node polls a constant number $k$ of randomly chosen neighbors in the interaction graph and changes its opinion as the majority opinion in the chosen neighborhood set, have been considered in the literature. The number of required interactions until convergence is achieved is shown to be $O(n \log n)$ for the complete interaction graph (both under sequential and parallel pull actions of nodes~\cite{BCNPST14,CG14}), $O(n \log n)$ for random regular interaction graphs when $k=2$ (under parallel pull actions of nodes~\cite{CER14}) and $O(n \log \log n)$ for Erd\H{o}s-R\'enyi random interaction graphs when $k=5$ (under additional constraints on the initial placement of opinions \cite{AD15}). Our work is also related to the voting model of DeGroot~\cite{DG74}, in which opinions spread through the network due to weighted interactions between agents. However, the DeGroot model associates weights of influence with particular agents, rather than with the opinions they hold, and consequently the obtained equations for propagation of opinion are inherently linear.

\paragraph{Discrete Lotka-Volterra dynamics and cyclic games.}

The continuous Lotka-Volterra dynamics, first defined in~\cite{L1910}, give rise to several discrete variants of so-called predator-prey models of interaction in a population, which differ essentially in the way the population size is maintained after the prey is attacked by the predator. Such models studied in the literature include the discrete May-Leonard model, in which the attacked agent (prey) disappears from the system, leaving behind a special state representing an empty niche, which can be later filled by another species, as well as the LV-type discrete dynamics studied in this work, in which the niche left by the prey is immediately filled by the species of the predator (cf.~\cite{SMJSRP14} for further generalizations of the framework). The LV-type model is particularly worthy of study due to its transient stability in a setting in which several species are in a cyclic predator-prey relation. Cyclic LV-type protocols have been consequently identified as a potential mechanism for describing and maintaining biodiversity, e.g., in bacterial colonies~\cite{KRFB02,KR04}. Cycles of length 3, in which type $1$ attacks type $2$, type $2$ attacks type $3$, and type $3$ attacks type $1$, form the basis of the best-known protocol, called rock-paper-scissors (RPS). The transient properties of RPS and related protocols, describing in particular the time until the system collapses to an absorbing state, have been studied in the statistical physics literature using a variety of experimental and analytical techniques\InJournal{ (mostly based on approximation with physical equations)}, under various scheduler models\InJournal{: the standard model of sequentially occurring random interactions, models with $\Theta(n)$ parallel encounters, models in which the discrete process is described by adding stochastic noise to a continuous-time process, etc}. The original analytical estimation method applied to RPS was based on approximation with the Fokker-Planck equation~\cite{RMF06}. A subsequent analysis of cyclic $3$- and $4$-species models using Khasminskii stochastic averaging can be found in~\cite{DF12}, together with a general discussion of the classification of dynamics according to convergence speed. A mean field approximation-based analysis of RPS was performed in~\cite{PK09}. All of these results provide a qualitative understanding of cyclic protocols, and at a quantitative level, provide evidence that the RPS protocol reaches an absorbing state after roughly $O(n^2)$ interactions. Finally, we remark that in a computer science setting,~\cite{CS08} made the connection between the Lotka-Volterra equations and (computational) population protocols, defining the class of so-called Linear Viral Protocols, in which the capacity of agents to participate in interactions changes over time.

\subsection{Model and Preliminaries}

We consider population protocols in the following setting. The population $V$ with $k$ types (species) is a set of $n$ agents, with each agent $v\in V$ assigned a state variable $s$, whose value at time $t$ is denoted $s_t(v) \in \{1,\ldots,k\}\equiv[k]$, describing its current type. The elements of $V$ are connected into an (undirected, connected) \emph{interaction graph} $G = (V,E)$. Agents assigned to type $i$ at time $t$, $1\leq i \leq k$, are called the \emph{population} of type $i$ at time $t$.

The population protocol $P$ is a probability distribution over $[k]^2$, taking values in $[k]^2$. In an execution of protocol $P$, at each time step $t=1,2,3,\ldots$, a scheduler daemon picks a pair of interacting agents $u, v\in V$ such that $(u,v) \in E$ u.a.r., and updates the state variables of these agents, sampling the pair $(s_{t+1}(u), s_{t+1}(v))$ according to the distribution $P(s_{t}(u), s_{t}(v))$. We will say that the population protocol is of the \emph{Lotka-Volterra type} (LV-type for short) if the state of the initiating agent (the predator) never changes during an interaction, and the state of the other agent (the prey) either remains unchanged or changes to that of the initiator, i.e., for any transition which occurs with non-zero probability, we have $s_{t+1}(u) = s_{t}(u)$ and $s_{t+1}(v) \in \{s_t(u),s_t(v)\}$.

For $i\in [k]$ and a fixed execution of protocol $P$, we will denote the size of the $i$-th population as $n_i(t) = |\{v \in V : s_t(v) = i\}|$, and its relative size as $x_i(t) = n_i(t)/n$. The set of states of all $n$ agents at time $t$ is referred to as the \emph{state} or \emph{configuration} of the system. When the interaction graph is the complete graph $K_n$, then we identify the state of the system with the vector $x(t)$. For $G=K_n$, the protocol $P$ defines a Markov chain on the set $X$ of possible states $x(t)$. We note that in this case, the size of the state space can be trivially bounded as $\Oh(n^k)$, i.e., is polynomial in $n$ for any fixed protocol.


In the sequel, any LV-type protocol $P$ will be identified with its $k\times k$ probability matrix $P$, such that for an interaction $(u,v)$, we have $s_{t+1}(v) = s_t(u)$ with probability $P_{s_t(u), s_t(v)}$, and $s_{t+1}(v) = s_t(v)$ with probability $1 - P_{s_t(u), s_t(v)}$. (Informally, we may write: ``$ij\to ii$ with probability $P_{ij}$''.) In general, matrix $P$ need not be skew-symmetric nor stochastic. We only assume that $P_{ii} =0$, for $1\leq i \leq k$, and that every type interacts in some way with at least one other type (for every $i$, $1\leq i \leq k$, there exists $j$, $1\leq i \leq k$, such that $P_{ij} > 0$ or $P_{ji} > 0$). We will denote the value of the minimal non-zero entry of matrix $P$ as $P_{\min}$. For every LV-type protocol, we construct the corresponding digraph $\vec D(P)$, whose vertex set is the set of types $[k]$, and an arc $(i,j)$ exists if $P_{i,j} >0$. We call the dynamics \emph{irreducible} if the digraph $\vec D(P)$  has no sources (i.e., there are no types without a predator, so each column of matrix $P$ has at least one non-zero entry) and is connected.

We remark that as $n\to\infty$, our random process converges to its (deterministic) limit continuous dynamics, given by the following set of first-order differential equations (a special case of the continuous Lotka-Volterra equations):
\begin{equation}\label{eq:lv}
\frac{d x_i(t)}{dt} = x_i(t) \sum_{j=1}^k \left[(P_{ij} - P_{ji})x_j(t)\right], \quad \text{for $1 \leq i \leq k$.}
\end{equation}
The above dynamics is non-linear and exhibits non-trivial properties in terms of limit behavior and stability measures such as the Lyapunov exponent (cf.\ \cite{HS98}). Our discrete population case with finite $n$ can be informally seen as a special form of ``noise'' introduced into the Lotka-Volterra equation~\eqref{eq:lv}.

In this paper, we also give our attention to two specific LV-type protocols:
\begin{itemize}
\item \emph{Rock-Paper-Scissors (RPS)} is the LV-type protocol with $k=3$ types (denoted $1, 2, 3$), whose probability matrix $P$ has the following non-zero entries: $P_{12} = P_{23} = P_{31} =1$.
\item \emph{Wolves-and-Sheep (WS)} is the LV-type protocol with $k=4$ types (denoted $X,Y,x,y$), whose probability matrix $P$ has the following non-zero entries: $P_{XY} = P_{Xx} = P_{YX} = P_{Yy} = 1$, $P_{Xy} = P_{Yx} = 1/2$.
\end{itemize}

Throughout the paper, we use the term ``with very high probability'' (w.v.h.p.) to denote events occurring with probability at least $1 - e^{-\Omega(\log^2 n)}$ and the term ``with high probability'' (w.h.p.) for events occurring with probability at least $1 - n^{-\Omega(1)}$. Thus, polynomially (resp., logarithmically) many events which individually hold w.v.h.p.\ (w.h.p.)\ also hold jointly w.v.h.p.\ (w.h.p.)\ by virtue of the union bound. In time and distance analysis, we will use the notation $\Otilde$ and $\Omegatilde$ to conceal poly-logarithmic factors ($\Otilde(f) = \Oh(f \mathrm{polylog}(n))$, $\Omegatilde(f) = \Omega(f /\mathrm{polylog}(n))$).

\section{Convergence of Discrete LV-type Protocols}\label{sec:convergence}

\InJournal{For any LV-type protocol, there exists a subset $X_s \subseteq X$ of stationary absorbing states, which remain unchanged under transformation $P$ (for any $x\in X_s$, applying any transformation governed by protocol $P$ to $x$ preserves $x$). Note that for any state $x \in X\setminus X_s$, there exists an execution of the system which reaches some absorbing state in a finite number of steps with strictly positive probability.
}

\InJournal{\subsection{Bound on Convergence Time under Complete Interactions}\label{sec:polynomial}}

We start by showing that any LV-type protocol on a population of size $n$ converges to an absorbing state in time $\Oh(poly(n))$, when there are no population constraints (the interaction graph is $K_n$).

\begin{theorem}[LV-type convergence for complete interactions]\label{thm:lvtime}
For any probability matrix $P$, there exists a constant $c$ such that the LV-type protocol defined by $P$ converges for the complete interaction graph to an absorbing state in $\Oh(n^c)$ steps, w.v.h.p.
\end{theorem}

Before proceeding with the proof, we introduce some auxiliary notation. For a fixed matrix $P$, we define the skew-symmetric \emph{nett interaction matrix} $A$ as $A = P - P^T$. Observe that $A_{ij} = P_{ij} - P_{ji}$ and equation~\eqref{eq:lv} describing the continuous dynamics now takes the simpler form:
\begin{equation}\label{eq:lvs}
\frac{d x_i}{dt} = x_i (A_i x),
\end{equation}\InJournal{
or in vector notation:
$$
\frac{d x}{dt} = x \odot (A x),
$$
}
where we treat $x$ as a column vector, and $A_i$ is the $i$-th row vector of matrix $A$\InJournal{, and $\odot$ denotes the element-by-element product of two column vectors} (cf.~\cite[Chapter 7]{HS98} for a more detailed exposition of the properties of this continuous dynamics).

For a fixed real vector $b \in \R^k$, which we will appropriately choose later, we define the potential $U$ of a system state $x$ as (compare~\cite[equation (5.3)]{HS98}):
$$
U(x) = \sum_{i=1}^k b_i \ln x_i,
$$
and by $U(t)$ we will mean $U(x(t))$. Observe that under evolution of the system given by the continuous dynamics~\eqref{eq:lvs}, we have:
$$
\frac{d U}{dt} = \sum_{i=1}^k b_i \frac{1}{x_i}\frac{dx_i}{dt} = \sum_{i=1}^k b_i (A_i x) = \sum_{j=1}^k \left( \left (\sum_{i=1}^k b_i A_{ij} \right) x_j \right) = b^T A x.
$$
We define vector $b$ as follows:
\begin{enumerate}
\item[(i)] if there exists a non-zero vector $b \geq 0$, such that $b^T A = 0$, choose $b$ as any such vector with $\|b\|_\infty = 1$.
\item[(ii)] otherwise, choose $b$ as any vector satisfying $b^T A  > 0$, with $\|b\|_\infty = 1$.
\end{enumerate}
The completeness of the above definition follows from a basic theorem of linear optimization, known as the ``no arbitrage theorem'' in financial mathematics (cf. also~\cite[proof of Thm.~5.2.1]{HS98})\InJournal{\footnote{For completeness, we provide a short version of the argument. If (i) does not hold, then the non-negative cone $\R^k_{\geq 0}$ intersects with the kernel $\ker A^T$ only at point 0. It follows that there must exist a strictly positive vector $h\in \R^k_+$ orthogonal to the kernel, $h \perp \ker (A^T)$. Since the orthogonal complement of $\ker(A^T)$ is the column space of matrix $A$, which is equal to the row space of $A$ (since $A = -A^T$), the equation $b^T A = h$ admits a solution $b$, and it is possible to satisfy condition (ii).}}.
\begin{proof}[Proof of Theorem~\ref{thm:lvtime}\InConference{, sketch}]
Observe that by the definition of the LV-type process, if a certain type $i$ has been eliminated by time $t$ ($x_i(t) =0$), then it will never reappear ($x_i(\tau)=0$, for all $\tau \geq t$). We will now show that the number of non-zero values of $x_i$, $1\leq i \leq k$, is reduced by at least one within a polynomial number of steps, w.h.p. Note that this is sufficient to obtain the claim of the theorem, since we may iterate the argument, each time restricting the definition of the dynamics and the matrix $A$ to those of the $k$ types which are non-empty. In the rest of the proof, we will be assuming w.l.o.g.\ that $x_i > 0$, for all $i$. We will also be assuming that the dynamics is irreducible; otherwise, if digraph $\vec D(P)$ is disconnected, we can consider each of the weakly connected components separately, and if any of the weakly connected components has a source, then one can easily show that all the prey of this source is eliminated in a polynomial number of steps.

The main part of the proof is contained in the following claim.

\smallskip
\emph{Claim. For any irreducible LV-type protocol, there exists a constant $n_{\min}$, such that for any $n>0$, for any initialization of the protocol with $n$ agents, w.v.h.p.\ there exists a time step $T \in \Oh(\poly(n))$ in which $n_i(T) < n_{\min}$, for some type $i$, $1\leq i \leq k$.}

\smallskip

\InConference{The proof of the claim proceeds by a careful analysis of the change of potential $U$ in time. Depending on case $(i)$ and $(ii)$ in the definition of vector $b$, we provide an absolute lower bound on the expectation $\E(\delta(t)|x(1),\ldots,x(t-1))$ for the stochastic process $\delta(t) = U(t+1) - U(t)$, given that $n_i(T) \geq n_{\min}$ for all types $i$. We then perform a super/sub-martingale analysis for the deviation of $\delta(t)$ from this expected change of potential for the two cases $(i)$ and $(ii)$, applying Azuma's inequality to bound the number of steps of the process until we reach $n_i(T) < n_{\min}$, for some type $i$. The details are deferred to the Appendix.}

\InJournal{
\emph{Proof (of claim).} For any $t>0$, we consider the random variable $\delta$ representing the change of the value of the potential function $U$ between time steps $t$ and $t+1$:
\begin{equation}\label{eq:deltadef}
\delta(t) = U(t+1) - U(t).
\end{equation}
We start by remarking that in a single time step, the change of potential is restricted to a single interaction between some pair $(i,j)$, and thus bounded by the following expression when $n_i \geq n_{\min}$, for all $1\leq i\leq k$:
\begin{align}
|\delta(t)| \leq& \max_{1\leq i,j\leq k}  \left| b_i \left(\ln \frac{n_i +1}{n} - \ln \frac{n_i}{n}\right) + b_j \left(\ln \frac{n_j -1}{n} - \ln \frac{n_j}{n}\right)\right| < \nonumber\\
<& 2\max_{1\leq i\leq k} |\ln(n_i) - \ln(n_i - 1)| < \frac{4}{n_{\min}}.\label{eq:deltabound}
\end{align}

Since the probability of the interaction of the form $(i,j)$ in the $(t+1)$-st step is $\frac{n_i n_j}{n^2}$, and this interaction increases $n_i$ by $1$ and decreases $n_j$ by $1$ with probability $P_{ij}$, the expectation $\E \delta(t)$ takes the form:
\begin{align*}
\E\delta(t) &= \sum_{i=1}^k \sum_{j=1}^k \frac{n_i n_j}{n^2} P_{ij} \left( b_i \left(\ln \frac{n_i +1}{n} - \ln \frac{n_i}{n}\right) + b_j \left(\ln \frac{n_j -1}{n} - \ln \frac{n_j}{n}\right)\right)=\\
& = \frac{1}{n^2} \sum_{i=1}^k  \sum_{j=1}^k b_i n_i n_j \left( P_{ij} \left(\ln (n_i+1) - \ln(n_i)\right) + P_{ji} \left(\ln (n_i-1) - \ln(n_i)\right)\right).
\end{align*}
We now introduce the following notation for Taylor series remainders in the expansion with respect to $1/x$: $\ln(x+1) - \ln(x) = \frac{1}{x} - \frac{1}{2 x^2} (1 - \frac{\lnc(x)}{x})$ and $\ln(x) - \ln(x-1) = \frac{1}{x} + \frac{1}{2 x^2} ( 1+ \frac{\lnf(x)}{x})$, where $0 \leq \lnc(x) \leq 2/3$ and $0\leq \lnf(x) \leq 2/3$, for every  $x>1$. Next, we obtain after elementary transformations:
\begin{align}
\E\delta(t) &= \frac{1}{n^2} \sum_{i=1}^k  \sum_{j=1}^k b_i A_{ij}n_j -
\frac{1}{2 n^2} \sum_{i=1}^k  \sum_{j=1}^k b_i \frac{n_j}{n_i} \left( P_{ij} \left(1 - \frac{\lnc(n_i)}{n_i}\right)  +  P_{ji} \left(1 + \frac{\lnf(n_i)}{n_i}\right) \right) \nonumber\\
= \frac{1}{n^2}& \sum_{i=1}^k  \sum_{j=1}^k b_i n_j \left( A_{ij}-
\frac{1}{2n_i} \left( P_{ij} \left(1 - \frac{\lnc(n_i)}{n_i}\right)  +  P_{ji} \left(1 + \frac{\lnf(n_i)}{n_i}\right) \right)\right).\label{eq:expdelta}
\end{align}

We now use~\eqref{eq:expdelta} to bound the expectation $\E\delta(t)$ depending on whether condition (i) or (ii) is satisfied by vector $b$.

\begin{itemize}
\item Case (i): Let $b_i\geq 0$ for all $1\leq i \leq k$, $b^T A=0$, and $b_{\alpha} = 1$ for some $\alpha$, $1\leq \alpha \leq k$. We have $\sum_{i=1}^k \sum_{j=1}^k b_i A_{ij}n_j =0$. Consequently, assuming $n_i \geq 2$, for all $1\leq i \leq k$, we may write:

\begin{align}
\E\delta(t) &= - \frac{1}{2n^2} \sum_{i=1}^k  \sum_{j=1}^k b_i \frac{n_j}{n_i} \left( P_{ij} \left(1 - \frac{\lnc(n_i)}{n_i}\right)  +  P_{ji} \left(1 + \frac{\lnf(n_i)}{n_i}\right) \right)  \label{eq:expdeltacasei}\\
\leq& - \frac{1}{3n^2} \sum_{i=1}^k  \sum_{j=1}^k b_i \frac{n_j}{n_i} \left( P_{ij} + P_{ji}\right) \leq
- \frac{b_\alpha}{3n^2 n_\alpha} \sum_{j=1}^k n_j \left( P_{\alpha j} + P_{j \alpha}\right) \leq
- \frac{P_{\min}}{3n^3}.\nonumber
\end{align}
In the above transformations, we took into account the fact that all $b_i$ are non-negative, that $0 \leq r^+(n_i), r^-(n_i) \leq 2/3$ for  $n_i >1$, and that by convention, the nonnegative-valued matrix $P$ has a non-zero entry in every row or corresponding column.

We now define a random variable $\delta'(t)$, given as $\delta'(t) = \delta(t)+\frac{P_{\min}}{3n^3}$ in every step of the process where $n_i \geq 2$, for all $1\leq i \leq k$, and set $\delta'(t) = 0$ for all other steps of the process. Clearly, $X_t = \sum_{\tau =0}^t \delta'(\tau)$ is a super-martingale in $t$, satisfying the bound $|X_t - X_{t-1}| \leq |\delta(t)| < 2$ by~\eqref{eq:deltabound}. By Azuma's inequality, $\Pr[X_t < \frac{P_{\min}}{6n^3} t] \geq 1 - \exp\left[-\left(\frac{P_{\min}}{6n^3} t\right)^2\frac{1}{8 t}\right] = 1 - \Oh(e^{-n})$, for any sufficiently large value of $t = \Omega(n^7)$. Observe, however, that the events ``$\delta'(\tau) = \delta(\tau)+\frac{P_{\min}}{3n^3}$ for all $\tau \leq t$'' and ``$X_t < \frac{P_{\min}}{6n^3} t$'' cannot hold simultaneously, since otherwise we have $U(0) - U(t+1) = \frac{P_{\min}}{3n^3}t - X_t > \frac{P_{\min}}{6n^3}t = \Theta (n^4)$. On the other hand, $x_i \geq 1/n$ and so every potential $U(\tau)$ satisfies:
$$
|U(\tau)| \leq k \ln n = \Oh(\ln n),
$$
hence such a potential difference is impossible. It follows that the probability that we have $n_i \leq 1$, for some $1\leq i \leq k$, for all except at most $\Theta(n^7)$ steps of the process, is $1 - \Oh(e^{-n})$.

\item Case (ii): Let $h = b^T A$, with $h_i \geq h_{\min} > 0$ for all $1\leq i \leq k$, and let $\|b\|_\infty = 1$. Assuming $n_i \geq n_{\min
} = \max\{2, \frac{2k}{h_{\min}}\}$, for $1\leq i \leq k$, we now have:
\begin{align*}
\E\delta(t) &= \frac{1}{n^2} \sum_{i=1}^k  \left( h_i n_i - \sum_{j=1}^k b_i \frac{n_j}{2n_i} \left( P_{ij} \left(1 - \frac{\lnc(n_i)}{n_i}\right)  +  P_{ji} \left(1 + \frac{\lnf(n_i)}{n_i}\right) \right)\right) \\
\geq\frac{1}{n^2}& \sum_{i=1}^k  \left( h_{\min} n_i - \sum_{j=1}^k \frac{n_j}{2n_{\min}} \left( \frac{4}{3} + \frac{4}{3} \right)\right) = \frac{h_{\min}}{ n}\left(1 - \frac{4k}{3n_{\min} h_{\min}}\right) \geq \frac{h_{\min}}{3 n}.
\end{align*}
Similarly as in the case of condition (i), we define a random variable $\delta'(t)$, given as $\delta'(t) = \delta(t)-\frac{h_{\min}}{6n}$ in every step of the process where $n_i \geq n_{min}$, for all $1\leq i \leq k$, and set $\delta'(t) = 0$ for all other steps of the process. This time, $X_t = \sum_{\tau =0}^t \delta'(\tau)$ is a sub-martingale in $t$, satisfying the bound $|X_t - X_{t-1}| \leq |\delta(t)| < \frac{4}{n_{\min}}$ by~\eqref{eq:deltabound}. Applying Azuma's inequality once again, we obtain $\Pr[X_t > -\frac{h_{\min}}{6n} t] \geq 1 - \exp\left[-\left(\frac{h_{\min}}{6n} t\right)^2\frac{1}{2 (4/n_{\min}^2) t}\right] = 1 - \Oh(e^{-n})$, for a sufficiently large value of $t = \Theta(n^2)$. By a similar potential difference argument as before, we obtain that the event $n_i < n_{\min}$, for some $1\leq i \leq k$, occurs within the first $\Theta(n^2)$ steps with probability $1 - \Oh(e^{-n})$.
\end{itemize}
This completes the proof of the claim.
\smallskip
}

To obtain the claim of the theorem, we now need to notice that whenever the population of some type drops below a constant threshold $n_{\min}$, the probability that the population is eliminated completely within the next $\Oh(1)$ steps of the irreducible protocol is polynomially large in $n$. \InJournal{Indeed, suppose we have $n_i(t) \leq n_{\min}$. Then, by the irreducibility of the protocol, there must exist an active predator $j$ of type $i$, $1\leq j \leq k$, with $P_{ji} \geq P_{\min}$ and $n_j \geq 1$ by assumption. The probability that species $i$ is eliminated in the next $n_i(t)$ steps can be lower-bounded by the probability of the occurrence of elimination of a representative of population $i$ by a representative of population $j$ in each of those steps, $(\frac{P_{\min}}{n^2})^{n_{\min}} = 1 / \Oh(n^{2 n_{\min}})$.} Overall, after at most $\Oh(n^{2 n_{\min}+1})$ occurrences of the event ``there exists $1\leq i\leq k$ such that $n_i < n_{\min}$'', each of which takes place every polynomial number of steps w.v.h.p.\ by the Claim, one of the species will have been eliminated completely w.v.h.p., which gives the claim of the theorem.
\end{proof}

We remark that the key property of LV-type protocols (i.e., that prey is changed to the same type as the predator) is essential in guaranteeing the polynomial time of reaching an absorbing state. For example, consider the following (non-LV) population protocol on 3 types $a$,$b$, and $c$: $\{ab \to ac, ac \to aa, ca \to cb, ba \to bb, bb \to ba\}$, where all the transitions occur with probability $1$. This protocol always preserves the type of the predator, but sometimes sets the third type as the type of the prey. For this protocol, for any initialization such that $n_a(0) = \Omega(n), n_b(0) = \Omega(n)$, the expected time to reach an absorbing state is exponential in $n$.

\InJournal{\subsection{Rock-Paper-Scissors with the Star Interaction Graph}\label{sec:star}}
It turns out that for LV-type protocols, the convergence time may become exponential when the interaction graph is not complete. Whereas all LV-type protocols with $2$ species (e.g., the game of life-and-death~\cite{HS98}) converge in polynomial time to an absorbing state for any interaction graph, this is no longer true when the number of species is at least $3$. We observe this for the rock-paper-scissors (RPS) protocol on the star.

\InJournal{
\begin{definition}\label{def:rps}
The \emph{Rock-Paper-Scissors (RPS) protocol} is the LV-type protocol with $k=3$ populations, given by the following matrix: $P = \left[ \begin{array}{rrr}0 & 1 & 0 \\ 0 & 0 & 1 \\ 1 & 0 & 0\end{array}\right]$.
\end{definition}
}

\begin{theorem}[RPS convergence on the star]
The RPS protocol with a $K_{1,n}$ interaction graph, initialized so that initially each type has at least $n_{\min} \ge n/3 - n/200$ agents, reaches the absorbing state in expected time $T_{abs} \ge e^{n^{\Omega(1)}}$.
\end{theorem}
\InJournal{
\begin{proof}
In this section we will fix $G = (V,E)$ as a star on $n+1$ vertices, or equivalently full bipartite graph $K_{1,n}$:
$V = \{0,1,\ldots,n\}, E = \{ (0,1), (0,2), \ldots, (0,n)\}$.
Full state of the system is encoded by a quadruple $S_t = (s_t(0), n_1(t), n_2(t), n_3(t))$ (for clarity of notation, we do not count the state of a middle vertex towards any of $n_i(t)$).

Thus, the transition probabilities of system are given as follow (assuming $s_t(0) = 1$ w.l.o.g.):
\begin{align*}S_{t+1} =& (s_{t+1}(0), n_1(t+1), n_2(t+1), n_3(t+1)) =\\
=&
\begin{cases}
(1, n_1(t), n_2(t), n_3(t)) & \text{ with probability } \frac{n_1(t)}{n}\\
(1, n_1(t)+1, n_2(t)-1, n_3(t)) & \text{ with probability } \frac{n_2(t)}{n}\\
(3, n_1(t), n_2(t), n_3(t)) & \text{ with probability } \frac{n_3(t)}{n}
\end{cases}
\end{align*}

We define a sequence of timesteps of RPS process, that marks changes of $s_t(0)$. More specifically, we define a sequence of timesteps, where $t_0 = 0$, and $t_{i+1}$ is the smallest value larger than $t_i$ such that $s_{t_{i+1}}(0) \not= s_{t_i}(0)$.

Let us denote by $d(t) = n_{\max}(t) -  n_{\min}(t)$ the difference between largest and smallest of the population.
We denote the potential $U(t) = n_1(t)n_2(t)n_3(t)$.

Let $T$ be an arbitrary timestep such that $U(T) \le 0.037025 \cdot n^3$, while $U(T-1) \ge  0.037025 \cdot n^3$.
Since
$$\min_{x \in [-1/300,1/300]} (1/3+x)(1/3-1/200-x/2)(1/3-1/200+x/2) \ge 0.037025$$
then $d(T) \ge n/100$.

Let $\Delta T = n / 1000$. Thus, for any $t \in [T,T+\Delta T]$, $d(t) \in [\frac{8}{1000} n, \frac{12}{1000}n]$.
Thus we denote $t_i = T < t_{i+1} < ... < t_{i+\Delta i} \le T + \Delta T$.

With probability at least $(1 - e^{- n^{\Omega(1)} })$ we have $t_{j+1} - t_j \le n^{0.1}$. Thus below we assume this bound on length of this gap, and at the end of this analysis we will have to acknowledge the failure probability of $e^{- n^{\Omega(1)} }$.

Thus, $\Delta i \ge n^{0.9}$.
Let us take arbitrary $t_j$ such that $T \le t_j \le t_{j+3} \le T + \Delta T$. W.l.o.g. we can assume that $S_{t_j} = 1$. Denote $\Delta_1 = t_{j+1} - t_j$, $\Delta_2 = t_{j+2} - t_{j+1}$ and $\Delta_3 = t_{j+3}$.
Thus, the state at $t_{j+3}$ fulfills:
$$(n_1(t_{j+3}),n_2(t_{j+3}),n_3(t_{j+3})) = (n_1(t_{j}),n_2(t_j),  n_3(t_j)) + (\Delta_1 - \Delta_2,\Delta_3 - \Delta_1, \Delta_2 - \Delta_3),$$
where:
$$\E[\Delta_1 | S_{t_j}] = \lambda_1 \frac{n_2(t_j)}{n_3(t_j)} $$
$$\E[\Delta_2 | S_{t_j}] = \lambda_2 \frac{n_1(t_j)}{n_2(t_j)} $$
$$\E[\Delta_3 | S_{t_j}] = \lambda_3 \frac{n_3(t_j)}{n_1(t_j)} $$
$$ \lambda_1,\lambda_2,\lambda_3 \in [1 - \bigo(n^{-0.9}), 1 + \bigo(n^{-0.9})]$$

We can bound the expected potential change:
\begin{align}
\label{eq:exp_pot}
&\E[ U(t_{j+3}) - U(t_j) | S_{t_j} ] = \lambda_2 n_1(t_j)^2  - \lambda_3 n_2(t_j) n_3(t_j) + \lambda_3 n_3(t_j)^2 -\\
&- \lambda_1 n_1(t_j) n_2(t_j) + \lambda_1 n_2(t_j)^2 - \lambda_2 n_1(t_j) n_3(t_j) + R(t_j)\nonumber
\end{align}
where
\begin{align*}R(t_j) &= n_1(t_j) \E[ (\Delta_3-\Delta_1)(\Delta_2-\Delta_3) | S_{t_j} ] + n_2(t_j) \E[ (\Delta_1-\Delta_2)(\Delta_2-\Delta_3) | S_{t_j} ] + \\
&+n_3(t_j) \E[ (\Delta_1-\Delta_2)(\Delta_3-\Delta_1) | S_{t_j} ] +
  \E[ (\Delta_1-\Delta_2)(\Delta_3-\Delta_1)(\Delta_2-\Delta_3) | S_{t_j} ]
\end{align*}
and satisfies
$$|R(t_j)| = \bigo(n).$$

\begin{lemma}
\label{lemma:pot_change}
$$\E[ U(t_{j+3}) - U(t_j) | S_{t_j} ] \ge \frac{1}{2} d(t_j)^2 (1-\bigo(n^{-0.9}))$$
\end{lemma}
\begin{proof}

Let us denote:
$\vec{u} = (\lambda_2 n_1(t_j), \lambda_1 n_2(t_j), \lambda_3 n_3(t_j))$, $\vec{v} = (n_1(t_j),n_2(t_j),n_3(t_j))$, $\vec{w} = (n_3(t_j),n_1(t_j),n_2(t_j))$.

We can rewrite \eqref{eq:exp_pot} to:
$$\E[ U(t_{j+3}) - U(t_j) | S_{t_j} ] = \vec{u} \circ (\vec{v} - \vec{w}) + R(t_j).$$

We denote $\measuredangle(\vec{u},\vec{v}) = \eps$, $\measuredangle(\vec{v}-\vec{w}, \vec{v}) = \alpha$ and $\measuredangle(\vec{v}-\vec{w}, \vec{u}) = \alpha'$.

Observe, that $|\vec{u} - \vec{v}| \le \bigo(n^{-0.9}) |\vec{v}|$. Thus, from $\sin(\eps) \le \frac{|\vec{u}-\vec{v}|}{|\vec{v}|}$, we have $ \eps = \bigo(n^{-0.9}) $.

Observe, that $|\vec{v}| = |\vec{w}|$ and that $|\vec{v}-\vec{w}| = \Theta(d(t_j))$, thus
$$\cos \alpha = \sin(\frac{1}{2} \measuredangle(\vec{v},\vec{w}) ) = 1/2\frac{ |\vec{v}-\vec{w}|}{|\vec{v}|} = \Theta(1)$$

Thus, since:
$$\vec{v} \circ (\vec{v}-\vec{w}) = \frac{1}{2}( (n_1(t_j)-n_2(t_j))^2 + (n_2(t_j)-n_3(t_j))^2 + (n_3(t_j)-n_1(t_j))^2 ),$$
we can bound:
\begin{align*}&\vec{u}\circ(\vec{v}-\vec{w}) = |\vec{u}| |\vec{v}-\vec{w}| \cos \alpha' = (\vec{v} \circ (\vec{v}-\vec{w})) \frac{|\vec{u}|}{|\vec{v}|} \frac{ \cos \alpha' }{\cos \alpha } \ge \\
&\ge \frac{1}{2} d(t_j)^2 (1-\bigo(n^{-0.9})) \frac{ \cos (\alpha+\eps) }{\cos \alpha } = \frac{1}{2} d(t_j)^2 (1-\bigo(n^{-0.9})) ( \cos \eps - \tan \alpha \sin \eps ) = \\
&= \frac{1}{2} d(t_j)^2 (1-\bigo(n^{-0.9})) ( 1 - \bigo(n^{-0.9})^2 - \Theta(1) \cdot \bigo(n^{-0.9}) ) = \frac{1}{2} d(t_j)^2 (1-\bigo(n^{-0.9}))
\end{align*}
which, since $R(t_j) = \bigo(n)$, gives the desired claim.
\end{proof}

We also observe following bound on potential change:
$$|U(t_{j+3}) - U(t_j)| \le 3 \cdot n^{0.1}\cdot \bigo(n) \cdot \bigo(n) = \bigo(n^{2.1})$$


We define a following submartingale:
$$\mathcal{W}(i) = U(t_{3i}) - i\cdot\frac{1}{2}\cdot\frac{1}{125^2}n^2 \cdot(1-\bigo(n^{-0.9}))$$
(being submartingale follows from application of Lemma~\ref{lemma:pot_change} and observation that $d(t_j) \ge \frac{1}{125}n$).

By application of Azuma's inequality:
$$P(\mathcal{W}(0) - \mathcal{W}(\frac{1}{3} \Delta i)) \ge \Theta(n^{2.75})) \le e^{ - \frac{\Theta(n^{5.5})}{\Delta i \cdot \bigo(n^{4.2})}} = e^{-\Omega(n^{0.3})}$$
which is equivalent to saying, that w.v.h.p.:
$$U(t_{\Delta i}) \ge U(t_0) - \Theta(n^{2.75}) + \Delta i \cdot \Theta(n^2) (1 - \bigo(n^{-0.9}))  > U(t_0).$$

Thus, taking into account original $e^{-n^{\Omega(1)}}$ probability of failure, we get a following:
for any timestep $T$ such that $U(T)$ falls below particular threshold, there exists w.v.h.p. a timestep $T'>T$ such that $U(T') > U(T)$. Thus the expected time to reach the absorbing state (which requires $U = 0$) is at least $e^{n^{\Omega(1)}}$.
\end{proof}
}

\section{The Wolves-and-Sheep (WS) Protocol}\label{sec:majority}

In this section, we investigate the dynamics of the Wolves-and-Sheep LV-type protocol, aiming at replicating dynamics of infection spreading for two different infections and two types of partial immunity to infections.
\InJournal{
\begin{definition}
The \emph{Wolves-and-Sheep} (WS) protocol is the LV-type protocol with $k=4$ populations (denoted $X,Y,x,y$),
given by the following matrix: $P = \left[ \begin{array}{rrrr}0 & 1 & 1 & \frac{1}{2} \\ 1 & 0 & \frac{1}{2} & 1 \\ 0 & 0 & 0 & 0 \\ 0 & 0 & 0 & 0\end{array}\right]$
\end{definition}
}
In the considered setting, initially almost all the population consists of types $x$ and $y$ (susceptible agents known as ``sheep''). A constant number of infected agents of types $X$ and $Y$ (the ``wolves'') are introduced into the population. Following the definition of the protocol, a wolf acting as a predator infects a sheep of a type denoted by the same lower-case with probability $1$, and a sheep of the opposite lower-case type with smaller probability ($1/2$). Thus, in the protocol, population $x$ of sheep has affinity towards $X$ (or resistance for $Y$), and population $y$ has affinity towards $Y$ (or resistance for $X$).

A real-world setting for such a protocol is the following. Initially, a population of users has smartphones of two different manufacturers (say, xPhones and yPhones). Simultaneously, each of the manufacturers introduces a new smartphone model into the market. Upon meeting someone with a new phone, a user will be convinced to upgrade to this model with some probability, which is higher if they are already a user of the older product of the same manufacturer. How will the balance of market share of the two manufacturers change after the whole market has adopted the new phone? It turns out that the market share of the new products grows exponentially in time, until almost all users have upgraded. However, the growth rate for the manufacturer with the initially (slightly) larger market share is larger, hence eventually, it dominates almost all of the market.

We note that in the definition of the WS protocol, we also add some random drift between the species $X$ and $Y$, which does not affect the nature of the process, but allows us to achieve an absorbing state in which eventually only the dominant type is represented.

\begin{theorem}[majority amplification by WS]
Let $n_X(0) = 1$, $n_Y(0) = 1$, $n_x(0) = \Theta(n)$ and $n_y(0) = \Theta(n)$, such that $\frac{n_x(0)}{n_y(0)} = \frac{1+\eps}{1-\eps}$ for some absolute constant $\eps > 0$. Then the system reaches the absorbing state with only population $X$, w.h.p.
\end{theorem}

\InJournal{
\begin{proof}

We will divide the analysis into three consecutive phases.

\paragraph{Phase I.} In the first steps of the process, we experience rather random behavior due to the small populations of the wolves. Whereas we would like to show that $n_X (t) \gg n_Y(t)$, this may initially not be the case: after $cn$ steps, for $c \ll \log n$, with probability roughly $e^{-c}$ we still have $n_X (t) = 1$, while $n_Y(t)$ satisfies exponential growth and in expectation we have $\E n_Y(t) = 2^{\Omega(c)}$. However, we can bound this behavior, showing that at some time $t_f = \Theta(n\log n)$, population $X$ has grown to $n_X(t_f) = \Omega(\log n)$, while $n_Y(t_f)$ is bounded by a small polynomial of $n$, w.h.p.

Specifically, in Phase I we run the process for $f = \log_2 (\alpha \ln n) - 4$ iterations, with the $i$-th iteration, $i=1,2,\ldots,f$ starting at time $t_{i-1}$ and ending at time $t_{i}$ (we assume $t_0 = 0$). The duration of the $i$-th iteration is $\delta t_i = t_i - t_{i-1} = \frac{\alpha  n \ln n}{2^i}$ time steps, with $\alpha = 0.01\eps$.

We start by upper-bounding the size of population $n_Y(t_f)$ at the end of the phase; note that $t_f < 2 \alpha  n \ln n$. We can trivially dominate the growth process for $n_Y(t)$ by a simpler unbounded growth process $n'_Y(t)$ given by the dependence: $n'_Y(t+1) = n'_Y(t) + 1$ with probability $\frac{n'_Y(t)}{n}$ and $n'_Y(t+1) = n'_Y(t)$, otherwise. For any $c_1 \log_2 n < s < \log_2 n -1$, where $c_1 >0$ is an arbitrarily small positive absolute constant, by an application of Chernoff bounds for the negative binomial distribution, we can lower-bound the length of the time interval during which $n'_Y$ grows from $2^s$ to $2^{s+1}$ as at least $n/4$, with probability $1 - n^{-\Omega(1)}$. Thus, we have at time $t_f$:
$$
n_Y(t_f) \leq n'_Y(t_f) \leq 2^{c_1 \log_2 n} 2^{\frac{t_f}{n/4}} < n^{c_1 + 8\alpha} < n^{10 \alpha},
$$
w.h.p, for a suitably chosen value of $c_1$. Likewise, we can show that $n_X(t) < n^{10 \alpha}$ and $n_Y(t) < n^{10 \alpha}$ for all $t\leq t_f$, w.h.p.

Next, conditioned on the above events, we lower-bound the size of $n_X(t_f)$, considering the value $n_X(t_i)$ at the end of each iteration. Note that the probability of an interaction between an individual from $X$ and an individual from $Y$ during Phase I can be upper-bounded as $\Oh(t_f \frac{n^{10\alpha}}{n} \frac{n^{10\alpha}}{n}) = \Otilde (n^{-1 + 20\alpha})$, thus we can assume that no such interaction occurs during Phase I, w.h.p. We now prove by induction that in each of the $f$ iterations of the phase, the value of $n_X$ increases by a multiplicative factor of $2$, w.h.p. By the inductive assumption, fix $1\leq i \leq f$ and let $n_X(t_{i-1}) \geq 2^{i-1}$. Since we have throughout the $i$-th iteration that $n_X(t) < n^{10 \alpha}$ and $n_Y(t) < n^{10 \alpha}$, it follows that $n_x(t) + n_y(t) \geq 1 - n^{20 \alpha} > \frac{1}{2}$. Thus, process $n_X$ during our $i$-th iteration dominates the growth process $n_X'$, such that $n_X(t_{i-1}) = 2^{i-1}$ and $n_X(t+1) = n_X(t) + 1$ with probability $\frac{ 2^{i-1}}{4n}$ and $n'_X(t+1) = n'_X(t)$, otherwise. Once again, applying the Chernoff bound for the negative binomial distribution, we can upper-bound w.h.p.\ the length of the time interval $\tau_i$ during which $n'_X$ grows from $2^{i-1}$ to $2^i$ as:
$$
\tau_i \leq \frac{4n}{2^{i-1}}(c_2\ln n + 2(2^i - 2^{i-1})) = \frac{(8c_2) n\ln n}{2^i} + 8n \leq \frac{(8 c_2+ \alpha/2) n \ln n}{2^i} < \frac{\alpha n \ln n}{2^i} = \delta t_i,
$$
where the first inequality holds with probability $1 - n^{\Omega(1)}$ for any arbitrarily small absolute constant $c_2 > 0$, and the last inequality holds when we choose $c_2 < \alpha /16$. In this way, we have obtained the inductive claim, and overall, we have $n_X(t_f) \geq 2^{f} > \frac{\alpha}{16}\ln  n$, w.h.p. By an analogous analysis, we obtain $n_Y(t_f) \geq 2^{f} > \frac{\alpha}{16}\ln  n$, w.h.p.

\paragraph{Phase II.} In the second phase, we start with populations of wolves $\frac{\alpha}{16}\ln  n < n_X(t_f) < n^{10 \alpha}$ and $\frac{\alpha}{16}\ln  n < n_Y(t_f) < n^{10 \alpha}$. We will show that both populations now experience exponential growth, but with a larger growth rate for population $X$. By the time the total population of wolves has grown to a level of $\Theta(n/\log n)$, we have that $n_X$ asymptotically outgrows $n_Y$.

Let $T$ be the first time step such that at least one of the populations of wolves exceeds the threshold $\frac{n}{\log n}$, i.e., $n_X(T) = \frac{n}{\log n}$ or $n_Y(T) = \frac{n}{\log n}$. Starting from time step $t_f$, we will divide time into intervals of duration $\beta n$, where $\beta > 0$ is a sufficiently small positive constant (we put $\beta = \frac{\eps}{128}$). We will now obtain w.h.p.\ bounds on the growth rate of the two processes during each such time interval, until time $T$. (Since the population of wolves treated as a whole undergoes exponential growth, we have $T = \Oh(n \log n)$ w.h.p., thus we will need to perform union bounds on only a small number of events associated with the $\Oh(\log n)$ time intervals until time $T$.)

Now, fix a time interval $I_t = [t, t + \beta n]$, for $t_f \leq t \leq T$. Within this interval, we construct an undirected graph $G_t = (V, E_t)$, where $V$ is the population, and $E_t$ is the set of all pairs of agents $\{i,j\}$, such that there exists an interaction between $i$ and $j$ during interval $I_t$. We observe that $G_t \sim G'(n,m)$ where $G'(n,m)$ is a variant of the Erd\H{o}s-R\'enyi random graph model where we choose $m = \beta n$ edges uniformly at random, discarding repeated edges. This corresponds to an average edge density of approximately $\frac{2\beta}{n}$. For $i\geq 1$, let $V = S_{\leq 2} \cup S_{>2}$, where $S_{\leq 2}$ and $\cup S_{>2}$ are subsets of agents belonging to connected components of $G_t$ having at most $2$ vertices and more than $2$ vertices, respectively. By performing a random graph analysis, we have w.h.p.:
\begin{align*}
|S_{\leq 2}| \geq& (1-o(1))\left(\left(1-\frac{2\beta}{n}\right)^{n-1} n + 2\frac{2\beta}{n}\left(1-\frac{2\beta}{n}\right)^{2n-3}{n \choose 2}\right) \geq \\
\geq& (1-2\beta) n + 2\beta n (1 - 4\beta) \geq n - 8\beta^2 n.
\end{align*}
Thus, we have $|S_{>2}|<  8\beta^2 n$, w.h.p. Denoting by $Y_{>2} = S_{>2}\cap Y(t)$, by independence of choice of sets $S_{>2}$ and $Y(t)$ and a simple Chernoff bound, we obtain $|Y_{>2}| < 16 \beta^2 n_Y(t)$, w.h.p. We will upper bound the size of set $Y(t + \beta n)$ as follows:
$$
n_Y(t + \beta n) \leq n_Y(t) + \Delta_Y[t,t+ \beta n] + |Y_{>2}|,
$$
where $\Delta_Y[t,t+ \beta n]$ is the number of individuals which joined $Y$ as a result of interactions having individuals from set $Y(t)$ as the initiator in the time interval $[t,t+ \beta n]$. The number of considered interactions can be upper-bounded as $(1 + c_3) \beta n_Y(t)$, for an arbitrarily small positive constant $c_3>0$, w.h.p. In each interaction involving an element of set $Y(t)$, the current size of $Y$ increases by $1$ if the interaction involves an element from $y$, and remains unchanged otherwise. We can thus upper-bound the value $\Delta_Y[t,t+ \beta n]$  through stochastic domination, running a process $Y'$ for $(1 + c_3) \beta n_Y(t)$ steps, starting from $Y'_0 = 0$ and increasing $Y'$ by $1$ with probability $p_Y$ in each step, where:
$$
p_Y = \left(\frac{1}{2} - \eps\right) + \frac{1}{2} \left(\frac{1}{2} - \eps\right) \leq \frac{3}{4} - \frac{\eps}{2}.
$$
Applying another simple Chernoff bound, we eventually obtain w.h.p.:
\begin{align*}
\Delta_Y[t,t+ \beta n] \leq& (1 + c_3) \beta n_Y(t) \cdot (1 + c_3)p_Y \leq\\
\leq& (1+3c_3) \left(\frac{3}{4} - \frac{\eps}{2}\right) \beta n_Y(t) \leq \left(\frac{3}{4} - \frac{\eps}{4}\right) \beta n_Y(t),
\end{align*}
for sufficiently small choice of constant $c_3$. Overall, we have w.h.p.
\begin{equation}
n_Y(t + \beta n) \leq \left( 1 + \left(\frac{3}{4} - \frac{\eps}{4} + 16\beta \right)\beta\right)n_Y(t) \leq \left( 1 + \left(\frac{3}{4} - \frac{\eps}{8}\right)\beta\right)n_Y(t). \label{eq:growthY}
\end{equation}
Now, we lower-bound the rate of growth of $n_X$ in the time interval $[t, t + \beta n]$. Since we have $n_Y \leq \frac{n}{\log n}$ throughout the considered interval, the number of interactions in the time interval $[t, t + \beta n]$ having an element from $X(t)$ as the prey can be upper-bounded as $c_4 \ln n$, w.h.p., for an arbitrarily small positive constant $c_4>0$. We denote by $X^*(t)$ the subsect of individuals from $X(t)$ which do not become prey during the considered time interval; we have $|X^*(t)| \geq n_X(t) - c_4 \ln n \geq (1-\frac{16c_4}{\alpha}) n_X(t)$, for a suitable choice of constant $c_4$. Now, we lower-bound the number of interactions involving an individual from $X^*$ as the initiator during the time interval $[t, t + \beta n]$ as $(1-c_5) \beta X^*(t)$, w.h.p., for an arbitrarily small constant $c_5 > 0$. Since we have $n_x \geq (1 + \eps)\frac{n}{2} - \frac{2 n}{\log n}$ and $n_y \geq (1 - \eps)\frac{n}{2} - \frac{2 n}{\log n}$ throughout the considered Phase II of the process, it follows that the probability of increasing $n_X$ by $1$ in a step in which an element from $X^*$ is the initiator is at least:
$$
p_X = \left(\frac{1}{2} + \eps - \frac{2}{\log n}\right) + \frac{1}{2} \left(\frac{1}{2} - \eps - \frac{2}{\log n}\right) \geq \frac{3}{4} + \frac{\eps}{2} - \frac{3}{\log n}.
$$
Overall, we get that the number of new elements joining $X$ in time interval $[t, t + \beta n]$  is at least $(1 -c_6)\beta \left (\frac{3}{4} + \frac{\eps}{2}\right) X^*(t)$, w.h.p., for an arbitrarily small positive constant $c_6 >0$. Of these elements, we likewise have that (an arbitrarily small) constant proportion will be lost due to interactions initiated by $Y$ within the interval $[t, t + \beta n]$, whereas the remaining ones will belong to $X(t + \beta n)$. Overall, we obtain w.h.p.:
\begin{align}
&n_X(t + \beta n) \geq X^*(t) + (1 -c_7)\beta \left (\frac{3}{4} + \frac{\eps}{2}\right) X^*(t) \geq \nonumber\\
&\left(1 + (1 -c_7)\beta \left (\frac{3}{4} + \frac{\eps}{2}\right)\right) \left(1-\frac{16c_4}{\alpha}\right) n_X(t) \geq\left( 1 + \left(\frac{3}{4} + \frac{3\eps}{8}\right)\beta\right)n_X(t)\label{eq:growthX}
\end{align}
for some suitable choice of arbitrarily small positive constant $c_7 > 0$ and of $c_4$.

Equations~\eqref{eq:growthY} and \eqref{eq:growthX} provide a separation of the growth rates of processes $X$ and $Y$.

Observe, that for small enough $\beta$, we have (since $\frac{3/4+ \varepsilon3/8}{3/4 - \varepsilon/8} < 1+\frac23 \varepsilon$):
$$\frac{n_X(t+\beta n)}{n_X(t)} \ge \left(\frac{n_Y(t+\beta n)}{n_Y(t)}\right)^{1+\frac23 \varepsilon}.$$

Thus, we eventually obtain at the end of Phase II, w.h.p.:
$n_X(T) = \frac{n}{\log n}$, $n_Y(T) = \bigo( n^{0.1 \varepsilon + 1/(1+2/3 \varepsilon)} ) = \bigo( n^{1 - \varepsilon \frac{3}{10}})$ for $\varepsilon \le 1$.

\paragraph{Phase III.}
We are now in a situation that, w.h.p., in some timestep $T_0$, $n_X(T_0) \ge \frac{n}{\log n}$, $n_Y(T_0) = \bigo(n^{\kappa})$ for some constant $\frac{3}{4} \le \kappa < 1$. We can safely assume, that $n_Y(T_0) = \Theta(n^{\kappa})$, since by a simple coupling argument, swapping some of $x$ and $y$ for $Y$ cannot help with reaching absorbing state of only $X$.

We will analyse the speed of growth of both $n_X(t)$ and $n_Y(t)$.

\begin{claim}
\label{lem:12power}
\textit{There is increasing sequence of steps $T_0,T_1,T_2,\ldots,T_j$ such that $T_j$ is the first element of this sequence satisfying $n_x(T_j)+n_y(T_j) < n^{3/4}$, such that, w.h.p.:
$\frac{n_Y(T_i)}{n_Y(T_0)} \le \left(\frac{n_X(T_i)}{n_X(T_0)}\right)^{12}$, and that $n_X(T_i) =  \Thetatilde(n_X(T_0))$. }
 \end{claim}
\textit{Proof.}
We proceed by induction on $i$:

We assume we have a moment $T_i \ge T_0$. From the inductive assumption: $n_Y(T_i) = \Thetatilde(n_Y(T_0)) = \Thetatilde(n^\kappa)$, and that $n_x(T_i)+n_y(T_i) = \Omega(n^{3/4})$.
We pick as $T_{i+1}$ a timestep such that in timesteps $\{T_i,\ldots,T_{i+1}\}$ there are $n_Y(T_i))$ events of $X$ or $Y$ attacking $x$ or $y$.

Since $n_x(t) + n_y(t) \ge \frac12 n^{3/4}$, at each step with probability at least $(\frac12-o(1)) n^{-1/4}$ there is attack of $X$ or $Y$ on $x$ or $y$. Thus, by Chernoff bound, following is w.h.p.:
$$T_{i+1} - T_i = 4 n^{1/4} n_Y(T_i)$$

Thus, the attacks of $X$ on $Y$ and of $Y$ on $X$ (a random walk, so to speak) accounts, w.h.p., for a change of population in the timesteps $T_i,\ldots,T_{i+1}$ (denoted as $\Delta_i$):
$$ \Delta_i \le (T_{i+1} - T_i) ^ {2/3} = \Otilde( n^{\kappa} \cdot n^{1/4} )^{2/3} = o(n^{\kappa}) = o(n_Y(T_i))$$

Observe that for any $T_i \le t \le T_{i+1}$:
$$n_Y(t) \le 2n_Y(T_i) + \Delta_i \le 2.01n_Y(T_i),$$
$$n_Y(t) \ge n_Y(T_i) - \Delta_i \ge 0.99 n_Y(T_i),$$
$$n_X(t) \le n_X(T_i) + n_Y(T_i) + \Delta_i \le 1.01 n_X(T_i),$$
$$n_X(t) \ge n_X(T_i) - \Delta_i \ge 0.99 n_X(T_i).$$

Thus, if in the timestep $t$ one of the members of populations  $x$ or $y$ is drawn as a second participant and one of $X$ or $Y$ is drawn as a first participant, we have increase in population of $X$ with probability at least $(1-p_i)/2$ and increase in population of $Y$ with probability at most $p_i$, for
$$p_i \le \frac{n_Y(t)}{n_X(t)+n_Y(t)} \le 2.03 \cdot \frac{n_Y(T_i)}{n_X(T_i)}.$$

So, by Chernoff bound, we have, w.h.p:
\begin{equation}\label{eq:bla1}n_Y(T_{i+1}) \le n_Y(T_i) + n_Y(T_i) \cdot 1.1 \cdot p_i + \Delta_i \le n_Y(T_i) \cdot (1+3\frac{n_Y(T_i)}{n_X(T_i)}) \end{equation}
\begin{eqnarray}\label{eq:bla2}n_X(T_{i+1}) \ge& n_X(T_i) - \Delta_i + n_Y(T_i) \cdot 0.9 \cdot (1-p_i)/2 \ge \nonumber\\
\ge& n_X(T_i) \cdot \left(1 + \frac{1}{3}\frac{n_Y(T_i)}{n_X(T_i)} - \frac{1}{3}\left(\frac{n_Y(T_i)}{n_X(T_i)}\right)^2\right)\end{eqnarray}

From the inductive assumption:
\begin{equation}\label{eq:bla3}\frac{n_Y(T_{i})}{n_X(T_{i})} = \Thetatilde(n^{\kappa-1}). \end{equation}

Thus, we observe that (by \eqref{eq:bla2} and \eqref{eq:bla3}):
$$n_X(T_{i+1}) \ge n_X(T_i) \cdot \left(1 + \frac{1}{4}\frac{n_Y(T_i)}{n_X(T_i)}\right)$$
and using \eqref{eq:bla1}:
$$\frac{n_Y(T_{i+1})}{n_Y(T_{i})}  \le \left(\frac{n_X(T_{i+1})}{n_X(T_{i})} \right)^{12}$$
which completes the inductive proof of claim.

Thus, we can now complete the proof of main theorem.
By Lemma~\ref{lem:12power}, at moment $T_j$ we will have following:
$n_X(T_j) = \Theta(n)$, $n_Y(T_j) = \bigo( n^{\kappa} \cdot (\log n)^{12} )$, $n_x(T_j)+n_y(T_j) = \bigo(n^{3/4})$.
At this moment, we can as well assume that all of $x$ and $y$ becomes $Y$, giving us a win for $X$ in following life-death game of wolves with probability at least $1-\frac{n^{3/4}+n^{\kappa}(\log n)^{12}}{n}$, which holds w.h.p.
\end{proof}
}

\section{The Rock-Paper-Scissors (RPS) Protocol}\label{sec:rps}

In this section, our goal is to show that the RPS protocol reaches each of absorbing states with almost equal probability, given that the initial population of each species is linear (or slightly sub-linear) in $n$.
\InJournal{We recall that the probability matrix of the RPS protocol is given by Definition~\ref{def:rps}. The corresponding matrix $A$ is the following: $A=\left[ \begin{array}{rrr}0 & 1 & -1 \\ -1 & 0 & 1 \\ 1 & -1 & 0\end{array}\right]$.} The RPS protocol admits a cyclic symmetry of behavior with respect to its species. For each of the species $a\in \{1,2,3\}$, the relative change $\Delta x_a$ in population of this species in the given step can be expressed as:
\begin{equation}\label{eq:deltana}
\Delta x_a = \frac{1}{n} \cdot \Delta n_a =
\begin{cases}
+1/n, &\text{ with probability }x_{a}x_{a+1},\\
-1/n, &\text{ with probability }x_{a+2}x_{a},\\
0, &\text{ otherwise},
\end{cases}
\end{equation}
where the population of at most one species changes in every step. The indices of populations are always $1, 2,$ or $3$, and other values should be treated as $\modd 3$, in the given range. We also introduce the continuous dynamics $\cont x(t)$ corresponding to the RPS process, given for each species by the differential equation:
\begin{equation}\label{eq:contrps}
\frac{d \cont x_a}{dt} = \frac{\cont x_a}n (\cont x_{a+1} - \cont x_{a+2}),
\end{equation}
which corresponds precisely to the continuous dynamics \eqref{eq:lv}, up to an additional time-scaling factor of $n$ introduced for easier comparison with the discrete process.
In all further considerations, we set the potential $U$ used in the analysis as:
\InJournal{
$$U(x)= \sum_{i=1}^3 \ln x_i = \ln (x_1 x_2 x_3).$$
}
\InConference{
$U(x)= \ln x_1 + \ln x_2 + \ln x_3.$
}
Lines $U = \const$ correspond to orbits in the continuous setting~\eqref{eq:contrps}.


\begin{theorem}[coin-flip consensus property of RPS]\label{thm:rps}
For any state $x$ such that $x_a > n^{-0.002}$ for all $a\in\{1,2,3\}$, the probability of the system reaching any one of its three possible absorbing states is $\frac{1}{3} \pm \Otilde(n^{-0.05})$.
\end{theorem}

\InJournal{The rest of this section is devoted to the proof of Theorem~\ref{thm:rps}.}
The proof of the theorem relies on the observation that the discrete RPS protocol approximately follows the limit cycle (orbit) of its continuous version. More precisely, we will observe that for an appropriately chosen starting state $x(0) = (x_1, x_2, x_3)$ of the system, there is a time moment $t$ (corresponding to an approximate traversal of $1/3$ of the limit cycle) for which the state is given as $x(t) = (x_3 + \Delta x_3, x_1 + \Delta x_1, x_2 + \Delta x_2)$, with $\Delta x_i$ sufficiently small. We will then use this to observe that if the probability of reaching any fixed absorbing state $i$ from state $(x_1, x_2, x_3)$ is $p$ and of reaching absorbing state $i$ from state $(x_1 + \Delta x_1, x_2 + \Delta x_2, x_3 + \Delta x_3)$ is $p_\Delta$, then by cyclic symmetry of populations, the probability of reaching state $(i+1)\modd 3$ from state $x(t)$ is also $p_\Delta$. If $p \approx p_\Delta$, then state $x(0)$ leads to absorbing states $i$ and $(i+1)\modd 3$ with almost the same probability.

At an intuitive level, the main arguments of the proof are the following. To show that the probability of reaching an absorbing states are almost the same for points $x(0) = (x_1,x_2,x_3)$ and $y(0) = (x_3,x_1,x_2)$, we perform a coupling of walks starting from $x(0)$ and $y(0)$. Here, coupling of Markovian processes is understood in the usual sense (cf.~e.g.\cite{LPW06}), though it is worth noting that since we are interested only in reaching an absorbing state (and not measuring the number of steps after which such a state is reached), we can in some steps of the coupling decide to delay one of the walks, allowing the other to run, provided that each of the processes remains unbiased. For simplicity, suppose that a walk $x$ is located at a point at which all populations are of linear size in $N$ and the difference in size between the largest and smallest population is also linear in $n$ (e.g., $U(x) = -20$.) The behavior of the (undelayed) walk $x$ under our evolution in the next $t$ steps (for $t$ sufficiently small with respect to $n$) can be seen as a superposition of three types of motion:
\begin{enumerate}
\item Propagation along the trajectory $U(x) = \const$ at a speed approximately given by the evolution of the continuous process~\eqref{eq:contrps}. The Euclidean distance traversed in a single step is $\Theta(1/n)$, or $\Theta(t/n)$ over $t$ steps.
\item Random drift along the trajectory $U(x) = \const$ (slowing or accelerating with respect to the average speed). Over a short interval time of length $t$, this drift shifts the point by $\pm \Otilde(\sqrt t/n)$ along its trajectory.
\item Random drift orthogonal to the trajectory $U(x) = \const$. Over a short interval time of length $t$, this drift shifts the potential $U$ of the point by $\pm \Otilde(\sqrt t/n)$.
\end{enumerate}
The analysis of the process is somewhat technical, since the two types of random drift have slightly biased averages\InJournal{ (in particular, in Section~\ref{sec:polynomial} we took advantage of the fact that the orthogonal drift has an outward bias over very long intervals of time)}, the probabilities of different moves are changing over time, and the motion in different directions is not independent. The drift and the propagation speed also depend on the relation between the maximum and minimum of the sizes of the three populations, which change in time. Nevertheless, the random drifts of our process behaves closely enough to a combination of independent random walks that we can deal with them with a martingale-type analysis. \InJournal{We present the necessary tools in Subsection~\ref{sec:rpstech}, formulating them so that the claims hold even when the smallest population is sublinear. After that, in Subsection~\ref{sec:rpscoupling} we formalize our coupling framework and perform the claimed coupling of the points $x$ and $y$ (Lemma~\ref{lem:coupling}).}

\InJournal{
\subsection{Technical lemmas}\label{sec:rpstech}
}

To begin with, we recall the following simple property of a simple random walk on a line: a walk starting from point $0$ and proceeding for $T$ steps is confined to an interval of the form $[-\Otilde(\sqrt T), \Otilde(\sqrt T)]$ w.v.h.p, but is likely to hit all points at a distance of $o(\sqrt T)$ from $0$. By a Doob martingale analysis, we state a generalization of this property applicable to a wider class of processes (to the best of our knowledge requiring weaker assumptions than those in bounds given in the literature).

\begin{lemma}[concentration and anti-concentration]\label{lem:martingale}
Let $\eps > 0$ be an absolute constant. For $t = 1,\ldots,T$, let $X(t) = \sum_{\tau=1}^t \delta(\tau)$, be a random process which satisfies the following condition for some non-negative parameters $c, \alpha, \sigma$: If $|X(t-1)| \leq c\sqrt{T^{1+\eps}} + \alpha T$, then:
\begin{itemize}
\item[(i)] $\delta(t)$ is a bounded random variable, regardless of the history of the process: $$\left|\delta(t)\left|\right.\delta(t-1), \delta(t-2), \ldots, \delta(0)\right| \leq c,$$
\item[(ii)] The expectation of $\delta(t)$ is bounded, regardless of the history of the process: $$|\E (\delta(t) | \delta(t-1), \delta(t-2), \ldots, \delta(0))| \leq \alpha,$$
\item[(iii)] The standard deviation of $\delta(t)$ is lower-bounded, regardless of the history of the process: $$|\Var (\delta(t) | \delta(t-1), \delta(t-2), \ldots, \delta(0))| \geq \sigma^2.$$
\end{itemize}
Then:
\begin{itemize}
\item (Concentration bound) With probability at least $1 - e^{-\Omega(T^\eps)}$ we have:
\begin{equation}\label{eq:conc}
\forall_{t\in\{1,\ldots,T\}}\ |X(t)| \leq  c\sqrt{T^{1+\eps}} + \alpha T.
\end{equation}
\item (Anti-concentration bound)
For any $D \geq c$, with probability \\at least $1-82 \left(\frac{c^2}{\sigma^2}\frac{(D+\alpha T) \ln T}{\sigma \sqrt T}\right)^{2/3}$, we have:
\begin{equation}\label{eq:anticonc}
\exists_{t\in\{1,\ldots,T\}}\ X(t) \geq D.
\end{equation}
\end{itemize}
\end{lemma}
\InJournal{
\begin{proof}
We apply a succession of Doob-type martingale filters to the process under consideration. For a given realization of the process, let $X'(t) = \sum_{\tau =1}^t \delta'(\tau)$, where $\delta'(\tau) \equiv \delta(\tau) - \E(\delta(\tau) | \delta(\tau-1), \delta(\tau-2), \ldots, \delta(0))$ if $X'(t-1) \leq c\sqrt{T^{1+\eps}}$, and $\delta'(\tau) \equiv 0$, otherwise. In other words, we convert our process into a martingale by subtracting expectations of increments, and at the same time enclose $X'(t)$ between two ``barriers'' placed at $\pm (c\sqrt{T^{1+\eps}})$. If one of these barriers is crossed, $X'$ becomes fixed at this value thereafter.

For the concentration bound, we start by showing that $X'(t)$ has very small probability of crossing one of the barriers. Indeed, for all moments of time $t$ before the barriers are hit we have, $X(t) = X'(t) + \sum_{\tau=1}^t \E(\delta(\tau) | \delta(\tau-1), \delta(\tau-2), \ldots, \delta(0))$, and $|X'(t)| \leq c\sqrt{T^{1+\eps}}$. We immediately have by induction on time that then the conditions $|X(t)| \leq c\sqrt{T^{1+\eps}} + \alpha T$ and $|\E(\delta(\tau) | \delta(\tau-1), \delta(\tau-2), \ldots, \delta(0))|\leq \alpha$ are satisfied, and so $|\delta(t)| \leq c$ and also $|\delta'(t)|\leq 2c$. For any moment of time $t$ after one of the barriers is hit we have $\delta'(t) = 0$; thus, the condition $|\delta'(t)|\leq 2c$ holds throughout the time interval $t = 1,\ldots, T$. By Azuma's inequality for martingales, we have for any time $t$:
$$
\Pr[|X'(t)| \geq c\sqrt{T^{1+\eps}}] \leq 2 \exp\left(-\frac{c^2T^{1+\eps}}{2(2c)^2T}\right) = 2 e^{-T^\eps/8}.
$$
Thus, taking the union bound over all moments of time, the bound $[|X'(t)| < c\sqrt{T^{1+\eps}}]$ holds throughout the time interval $t = 1, \ldots, T$ with probability $p_1 = 1 - Te^{-T^\eps/8} = 1 - e^{\ln T-T^\eps /8}$. Thus, with probability $p_1$ the barriers are not reached, and we have $|\delta'(t)- \delta(\tau)| \leq \alpha$ throughout the entire time interval $[1,T]$. Thus, we obtain that the sought concentration bound for process $X$ holds with probability $p_1 = 1-e^{-\Omega (T^\eps)}$.

To show the anti-concentration bound, fix $D' > c$ arbitrarily. We will upper-bound the probability $p_2$ that the martingale $X'$ does not hit a barrier located at level $D'$ (i.e., $p_2 = \Pr[\forall_{t\in\{1,\ldots,T\}}\ X'(t) < D']$). Unlike probability $p_1$, which is very close to $1$, we expect $p_2$ to be close to $0$. Define martingale $X''(t) = \sum_{\tau=1}^T\delta''(\tau)$, where $\delta''(\tau) \equiv \delta'(\tau)$ if $X''(\tau-1) < D'$, and $\delta''(\tau) = 0$ otherwise. Observe that, conditioned on the events that process $X'$ does not hit either of its barriers at $\pm (c\sqrt{T^{1+\eps}})$ (which holds with probability at least $p_1$) and that process $X''$ does not hit its barrier at $D'$ (which holds with probability $p_2$), in a given realization of the process, $\delta''(t)$ and $\delta(t)$ have the same variance. Thus, by $(iii)$ we can write the following bound on the variance of $\delta''(t)$ which holds without conditioning:
$$
\Var(\delta''(t)) \geq p_2^* \sigma^2,
$$
where we introduce the notation: $p_2^* = p_2 - (1-p_1) = p_2 - e^{\ln T-T^\eps/8}$.
We now show that $p_2 \leq 16 \left(\frac{c^4}{\sigma^4}\frac{D'^2}{\sigma^2 T^{1-2\eps}}\right)^{1/3} + e^{\ln T-T^\eps/8}$. Suppose, to the contrary, that $p_2^* > 16 \left(\frac{c^4}{\sigma^4}\frac{D'^2}{\sigma^2 T^{1-2\eps}}\right)^{1/3}$. By the additivity of variance of a martingale over its elements, we have:
$$
\E(X''^2(T)) = \Var(X''(T)) = \sum_{t=1}^T \Var(\delta''(t)) \geq p_2^* \sigma^2 T .
$$
On the other hand, since $|X''(t)| \leq c \sqrt{T^{1+\epsilon}} + c \leq 2c \sqrt{T^{1+\epsilon}}$ by the same bound imposed on process $X'(t)$ through its definition, $t \in \{1,\ldots, T\}$, we also have:
$$
X''^2(T) \leq 4c^2 T^{1+\epsilon}.
$$
Since $X''^2(T)$ is an upper-bounded non-negative random variable, we can apply Markov's inequality to lower-bound its heavy tail:
\begin{equation}
\Pr\left[X''^2(T) \geq \frac{p_2^* \sigma^2 T}{4}\right] \geq \frac{1}{2} \frac{p_2^* \sigma^2 T}{4c^2 T^{1+\epsilon}} = \frac18 p_2^* \sigma^2 c^{-2} T^{-\epsilon}
\end{equation}
and by transforming the expression on the left-hand side, we get:
\begin{equation}\label{eq:heavy1}
\Pr\left[|X''(T)| \geq 2D'\frac{\sqrt{p_2^*} \sigma \sqrt{T}}{4D'}\right] \geq \frac18 p_2^* \sigma^2 c^{-2} T^{-\epsilon},
\end{equation}
Since the maximum value attained by $X''(T)$ is upper-bounded by $D'+c < 2D'$, and by assumption we have $p_2^* > 16 \left(\frac{c^4}{\sigma^4} \frac{D'^2}{\sigma^2 T^{1-2\eps}}\right)^{1/3} > \frac{16D'^2}{\sigma^2 T}$, we have $\frac{\sqrt{p_2^*} \sigma \sqrt{T}}{4D'}>1$ and $\Pr\left[X''(T) \geq 2D'\frac{\sqrt{p_2^*} \sigma \sqrt{T}}{4D'}\right]=0$. So, we can drop the absolute value in expression~\eqref{eq:heavy1}:
\begin{equation}\label{eq:heavy}
\Pr\left[X''(T) < -2D'\frac{\sqrt{p_2^*} \sigma \sqrt{T}}{4D'}\right] \geq \frac18 p_2^* \sigma^2 c^{-2} T^{-\epsilon}.
\end{equation}
On the other hand, $\E X''(T)=0$ and $X''(T) < 2D'$, so by Markov's inequality:
\begin{equation}\label{eq:light}
\Pr\left[X''(T) < -2D' \frac{\sqrt{p_2^*} \sigma \sqrt{T}}{4D'}\right] \leq \frac{1}{1+\frac{\sqrt{p_2^*} \sigma \sqrt{T}}{4D'}} <
\frac{4D'}{\sqrt{p_2^*} \sigma \sqrt{T}}.
\end{equation}
Combining the right-hand sides of expressions~\eqref{eq:light} and~\eqref{eq:heavy}, we obtain:
$$
\frac18 p_2^* \sigma^2 c^{-2} T^{-\epsilon} < \frac{4D'}{\sqrt{p_2^*} \sigma \sqrt{T}},
$$
and subsequently:
$$
p_2^* < \left(\frac{1024c^4}{\sigma^4} \frac{D'^2}{\sigma^2 T^{1-2\eps}}\right)^{1/3} < 16 \left(\frac{c^4}{\sigma^4} \frac{D'^2}{\sigma^2 T^{1-2\eps}}\right)^{1/3},
$$
a contradiction with our assumption. In this way, we have shown:
$$
p_2 \leq 16 \left(\frac{c^4}{\sigma^4}\frac{D'^2}{\sigma^2 T^{1-2\eps}}\right)^{1/3} + e^{\ln T-T^\eps/8}.
$$

Now, conditioned on the event that $X'(t)$ does not hit any of the points $\pm (c\sqrt{T^{1+\eps}})$ which holds with probability at least $p_1 = 1 - e^{\ln T-T^\eps/8}$, having $X'(t) \geq D'$ at some time $t$ implies $X(t) \geq D' - at \geq D' - \alpha T$. Thus, putting $D' = D + \alpha T$, we have:
$$
\Pr[\forall_{t \in \{1,\ldots,T} X(t) < D] \leq p_2 + p_1 = 16 \left(\frac{c^4}{\sigma^4}\frac{(D+\alpha T)^2}{\sigma^2 T^{1-2\eps}}\right)^{1/3} + 2e^{\ln T-T^\eps/8}.
$$
Choosing $\eps$ so that $T^{\eps}=11\ln T$, we can rewrite the above as:
\begin{align*}
\Pr[\forall_{t \in \{1,\ldots,T\}} X(t) < D] &\leq 16 \left(\frac{c^4}{\sigma^4}\frac{121 (D+\alpha T)^2 \ln^2 T}{\sigma^2 T}\right)^{1/3} + 2T^{-3/8} < \\
&< 82 \left(\frac{c^2}{\sigma^2}\frac{(D+\alpha T) \ln T}{\sigma \sqrt T}\right)^{2/3},
\end{align*}
which completes the proof of our anti-concentration bound.
\end{proof}
}
We now need to define a meaningful random process based on the RPS protocol, for which we could take advantage of Lemma~\ref{lem:martingale} in the proof of Theorem~\ref{thm:rps}. In fact, we will need to bound more than one such process to close the analysis. We first introduce two measures of distance of a pair of points $x^{(a)}$, $x^{(b)}$ in our state space:
\begin{itemize}
\item $d_U(x^{(a)}, x^{(b)}) = |U(x^{(a)}) - U(x^{(b)})|$,
\item $d_\infty(x^{(a)}, x^{(b)}) = \|x^{(a)} - x^{(b)}\|_\infty = \max \{|x^{(a)}_1 - x^{(b)}_1|, |x^{(a)}_2 - x^{(b)}_2|, |x^{(a)}_3 - x^{(b)}_3|\}$.
\end{itemize}
\InConference{
Taking into account Lemma~\ref{lem:martingale}, we now show a sequence of claims bounding the distances $d_U$ and $d_{\infty}$ of a process progressing under the discrete RPS protocol from that under the continuous dynamics~\eqref{eq:contrps}. For compactness, these are given here in the form of the following summary lemma.
\begin{lemma}\label{lem:convU}
Let $x(0)$ be a point in the state space of RPS with $U(x(0)) > -\gamma \ln n$ for some absolute constant $0 < \gamma < 1/6$, let $x(t)$ be the random variable representing the point reached after following the population protocol for $t$ steps starting from point $x(0)$.
Next, consider the process $\cont x(t)$ governed by the continuous RPS dynamics~\eqref{eq:contrps} with any starting point $\cont x(0)$ such that $d_\infty(\cont x(0),x(0)) \leq \Delta$, for some $\Delta >0$. Then, for sufficiently large $n$ and any integer $T>0$, the following claim holds:
\begin{itemize}
\item If $T \leq n^{5/3}$, then: $\forall_{t\in\{1,\ldots,T\}}\ d_U(x(t),x(0)) = \Otilde(T^{0.5}/n^{1-\gamma})$, w.v.h.p.
\item If $T \leq n^{2/3}$, then: $\forall_{t\in\{1,\ldots,T\}}\ d_\infty(x(t),\cont x(t)) = \Otilde(\Delta + T^{0.5}/n)$, w.v.h.p.
\item If $T \leq n^{5/3}$ and $-10 > U(x(0)) > - \gamma \ln n$ for some absolute constant $0 < \gamma < 1/6$, then: $\forall_{t\in\{1,\ldots,T\}}\ d_\infty(x(t),\cont x(t)) = \Otilde( (T n^{\gamma-2/3}+1) \cdot (\Delta + T^{0.5} n^{\gamma-1}))$, w.v.h.p.
\item If $n^{6\gamma} \leq T \leq n^{4/3-8\gamma}$, $\Delta \leq T^{0.5} n^{\gamma-1}$, and $-10 > U(x(0)) > - \gamma \ln n$ for some absolute constant $0 < \gamma < 1/6$, then there exists an integer time step $T' = (1 + o(1)) T$, such that $d_\infty (x(T'), \cont x(T)) = \Otilde(T^{0.5}n^{\gamma-1})$, w.v.h.p.
\end{itemize}
\end{lemma}
}

\InJournal{
\begin{lemma}\label{lem:convU}
Let $x(0)$ be a point in the state space of RPS with $U(x(0)) > -\gamma \ln n$ for some absolute constant $0 < \gamma < 1/6$, let $x(t)$ be the random variable representing the point reached after following the population protocol for $t$ steps starting from point $x(0)$. Then, for sufficiently large $n$ and $T \leq n^{5/3}$, the following claim holds: $\forall_{t\in\{1,\ldots,T\}}\ d_U(x(t),x(0)) = \Otilde(T^{0.5}/n^{1-\gamma})$, w.v.h.p.
\end{lemma}

\begin{proof}
To show the claim, we apply concentration bound~\eqref{eq:conc} to the process $X(t) = U(x(t)) - U(x(0))$. We have $X(t) = \sum_{\tau=1}^t \delta(\tau)$, with random variable $\delta(t)$ given by: $\delta(t) = U(x(t)) - U(x(t-1))$. 
Let $v = e^{U(x(0))} > e^{-\gamma \ln n} = n^{-\gamma}$. We verify that the first two conditions of Lemma~\ref{lem:martingale} hold for $\eps$ such that $T^\eps = \log^2 n$, $c = \frac{4}{nv}$, $\alpha = \frac{8}{n^2 v}$. Indeed, suppose that the condition from Lemma~\ref{lem:martingale}:
$$
|U(x(t-1)) - U(x(0))| = X(t-1) \leq c T^{0.5}\log n + \alpha T
$$
is satisfied at some time $t-1$, $l\leq t \leq T$. Under this assumption, we have for $T\leq n^{5/3}$ and sufficiently large $n$:
\begin{align*}
U(x(t-1)) \geq& U(x(0)) - \frac{4 T^{0.5}\log n}{nv} - \frac{8 T}{n^2v} > \\
>& U(x(0)) - n^{-1/6+\log \log n+\gamma} > U(x(0)) - \ln 2,
\end{align*}
hence the minimum size of a species at time $t-1$ is then lower-bounded by:
$$
e^{U(x(t-1))} > \frac{v}{2}.
$$
Following~\eqref{eq:deltabound}, we then have:
\begin{align*}
|\delta(t)| &\leq \max_{a\in\{1,2,3\}} \left| \left(\ln \frac{n_a(t-1) +1}{n} - \ln \frac{n_a(t-1)}{n}\right) + \left(\ln \frac{n_{a+1}(t-1) -1}{n} - \ln \frac{n_{a+1}(t-1)}{n}\right)\right| < \\
 &< \max_{a\in\{1,2,3\}} |\ln(n_a(t-1)) - \ln(n_a(t-1) - 1)| < \max_{a\in\{1,2,3\}} \frac{2}{n_a(t-1)} < \frac{2}{n e^{U(x(t-1))}} < \frac{4}{nv} =c.
\end{align*}
Next, we bound the expectation $\E \delta(t)$ by transforming the equality in~\eqref{eq:expdeltacasei} at point $x(t-1)$:
\begin{align}
\E\delta(t) &\leq - \frac{1}{3n^2} \sum_{a=1}^3\left (\frac{n_a(t-1)}{n_{a+1}(t-1)} +  \frac{n_{a+1}(t-1)}{n_{a}(t-1)}\right) < 0,\\
\E\delta(t) &\geq - \frac{2}{3n^2} \sum_{a=1}^3\left (\frac{n_a(t-1)}{n_{a+1}(t-1)} +  \frac{n_{a+1}(t-1)}{n_{a}(t-1)}\right)
\geq \nonumber\\
&\geq - \frac{4}{n \cdot \min_{a\in\{1,2,3\}} n_a(t-1)} > \frac{-8}{n^2v} = -\alpha.
\end{align}
Thus, we can apply concentration bound~\eqref{eq:conc}, obtaining that:
$$
\forall_{t\in\{1,\ldots,T\}}\
|U(x(t)) - U(x(0))| = |X(t)| \leq  \frac{4 T^{0.5}\log n}{nv} + \frac{8 T}{n^2v} = \Otilde(T^{0.5}/n^{1-\gamma})
$$
holds with very high probability $1-e^{-\Omega(T^{\eps})} = 1-e^{-\Omega(\log^2 n)}$.
\end{proof}

\begin{lemma}\label{lem:disttime}
For a pair of points $p,q$ lying on the same orbit $U(p) = U(q) = U < -6$, let $\cont t = \cont t(p,q) > 0$ denote that smallest time such that for the continuous process~\eqref{eq:contrps} originating at $p$ we have $\cont p(\cont t) =q$. If $q$ is separated from $p$ by at most $1/3$ of the orbit ($\cont t(p,q)\leq \cont t(p,p)/3$), then the following relations hold:
$$
0.05\frac{\cont t}{n} e^U \leq d_\infty (p,q) \leq \frac{\cont t}{n}.
$$
\end{lemma}

\begin{proof}
The bound $d_\infty (p,q) \leq \frac{\cont t}{n}$ follows immediately from the observation that for any $t\geq 0$ and infinitesimal time interval $dt$, we have for the continuous process~\eqref{eq:contrps}: $d_{\infty}(\cont p(t+ dt), \cont p(t)) \leq \frac{dt}{n}$. To show the lower bound on $d_\infty (p,q)$, observe that the linear distance covered in each infinitesimal time step can be lower-bounded as follows:
\begin{align*}
&\|p(t+ dt)- \cont p(t)\|_2 \geq d_{\infty}(\cont p(t+ dt), \cont p(t)) \geq \frac{1}{n}\min_{a\in \{1, 2, 3\}} p_a(t) \cdot \\
&\cdot (\max_{a\in \{1, 2, 3\}} p_a(t) - \min_{a\in \{1, 2, 3\}} p_a(t)) dt \geq \frac{e^U}{n} \left(\frac{1}{3} - e^{U/3}\right) dt \geq \frac{e^U}{6n} dt,
\end{align*}
where we took into account that $e^U = \prod_{a=1}^3 p_a(t)$ and that $U < -6$. Denoting by $l$ the length of the traversal from $p$ to $q$ along the orbit $U$, we obtain:
$$
l \geq \frac{e^U}{6n} \cont t
$$

The orbit $x_1 x_2 x_3 = e^U$ is the boundary of a convex subset of the plane $x_1 + x_2 + x_3 = 1$. Moreover, the entire fragment of the considered orbit between points $p$ and $q$ is contained within the triangle $(p,q,o)$, where $o$ is the intersection point of the straight lines in the considered plane, adjacent to the orbit at points $p$ and $q$, respectively. By the rotational symmetry of the orbit under rotation by $2\pi / 3$ in its plane around the point $(1/3, 1/3, 1/3)$ and the fact that $p$ and $q$ are apart from each other by at most $1/3$ of the orbit, it follows that $\angle(poq) \geq \pi /3$. By applying the law of cosines to triangle $(p,q,o)$ we obtain:
$$\|p-o\|_2 + \|q-o\|_2 \leq 2 \|p-q\|_2.$$
Now, we can write the following bounds:
$$
d_\infty (p,q) \geq \frac1{\sqrt 2} \|p-q\|_2 \geq \frac1{2\sqrt 2} (\|p-o\|_2 + \|o-q\|_2) > \frac l{2\sqrt 2} \geq \frac{e^U}{12\sqrt 2 n} \cont t > 0.05 \frac{e^U}{n} \cont t,
$$
which completes the proof.
\end{proof}

\begin{lemma}\label{lem:duinfty}
Let $x$ be a point in the state space and let $p$  be a point with minimum infinity norm distance to $x$ from among all points having the same potential as $p$, i.e., $p \in \arg\min_{p'} \{d_\infty(p', x(t_i)) : U(p') = U(p)\}$. If $U(x) < -3$ and $d_\infty (x, p) < e^{U(x)}$, then
$$d_\infty (x, p) < d_{U}(x,p).$$
\end{lemma}

\begin{proof}
Let $d = d_\infty (x, p)$. Let $x_{\min} = \min_{a\in\{1,2,3\}} x_a$ and let $x_{\max} = \max_{a\in\{1,2,3\}} x_a$. We note that $x_{\min}<e^{U(x)}<0.05$ and $x_{\max}\geq 1/3$. Assume first that $U(p) > U(x)$, and let $r$ be the point reached from $x$ by moving $d$ individuals from the largest population in $x$ to the smallest population in $x$.
By the definition of the potential, we have:
$$
U(r) \geq U(x) + \ln\frac{x_{\min}+d}{x_{\min}} - \ln\frac{x_{\max}+d}{x_{\max}} \geq U(x) + d\left(\frac{1}{2x_{\min}} - \frac{1}{x_{\max}}\right) > U(x) + d.
$$
Since  $d = d_\infty (x, p) = d_\infty (x, r)$, the straight line segment connecting $x$ and $r$ does not intersect with the potential orbit $U=U(p)$, hence $U(p)\geq U(r) > U(x) + d$, which completes the proof for the considered case.

When $U(p) < U(x)$, we define $r$ by moving $d$ individuals from the smallest population in $x$ to the largest population in $x$; we omit the details of the analogous argument.
\end{proof}

\begin{lemma}\label{lem:infdist}
Let $x(0)$ be a point in the state space of RPS and let $x(t)$ be the random variable representing the point reached after following the population protocol for $t$ steps starting from point $x(0)$. Then, for sufficiently large $n$ and any positive integer $T$, the following claim holds:
\begin{itemize}
\item[(i)] If $T \leq n^{2/3}$, then: $\forall_{t\in\{1,\ldots,T\}}\ d_\infty(x(t),\contz x(t)) = \Otilde(T^{0.5}/n)$, w.v.h.p.,
where $\contz x(t) = (\contz x_1(t), \contz x_2(t), \contz x_3(t))$ is the following linear approximation to the considered process:
 $$\contz x_a(t) = x_a(0)\left(1 + \frac{t}{n} ( x_{a+1}(0) -  x_{a+2}(0))\right),\text{ for $a\in\{1,2,3\}$}.$$
\end{itemize}
Next, consider the process $\cont x(t)$ governed by the continuous RPS dynamics~\eqref{eq:contrps} with any starting point $\cont x(0)$ such that $d_\infty(\cont x(0),x(0)) \leq \Delta$, for some $\Delta >0$. Then, the following claims hold:
\begin{itemize}
\item[(ii)] If $T \leq n^{2/3}$, then: $\forall_{t\in\{1,\ldots,T\}}\ d_\infty(x(t),\cont x(t)) = \Otilde(\Delta + T^{0.5}/n)$, w.v.h.p.
\item[(iii)] If $T \leq n^{5/3}$ and $-10 > U(x(0)) > - \gamma \ln n$ for some absolute constant $0 < \gamma < 1/6$, then: $\forall_{t\in\{1,\ldots,T\}}\ d_\infty(x(t),\cont x(t)) = \Otilde( (T n^{\gamma-2/3}+1) \cdot (\Delta + T^{0.5} n^{\gamma-1}))$, w.v.h.p.
\item[(iv)] If $n^{6\gamma} \leq T \leq n^{4/3-8\gamma}$, $\Delta \leq T^{0.5} n^{\gamma-1}$, and $-10 > U(x(0)) > - \gamma \ln n$ for some absolute constant $0 < \gamma < 1/6$, then there exists an integer time step $T' = (1 + o(1)) T$, such that $d_\infty (x(T'), \cont x(T)) = \Otilde(T^{0.5}n^{\gamma-1})$, w.v.h.p.
\end{itemize}
\end{lemma}

\begin{proof}
\begin{figure}
\centering
\ifpdf
\includegraphics[height=6cm,trim= 3cm 5.5cm 3cm 4cm,clip]{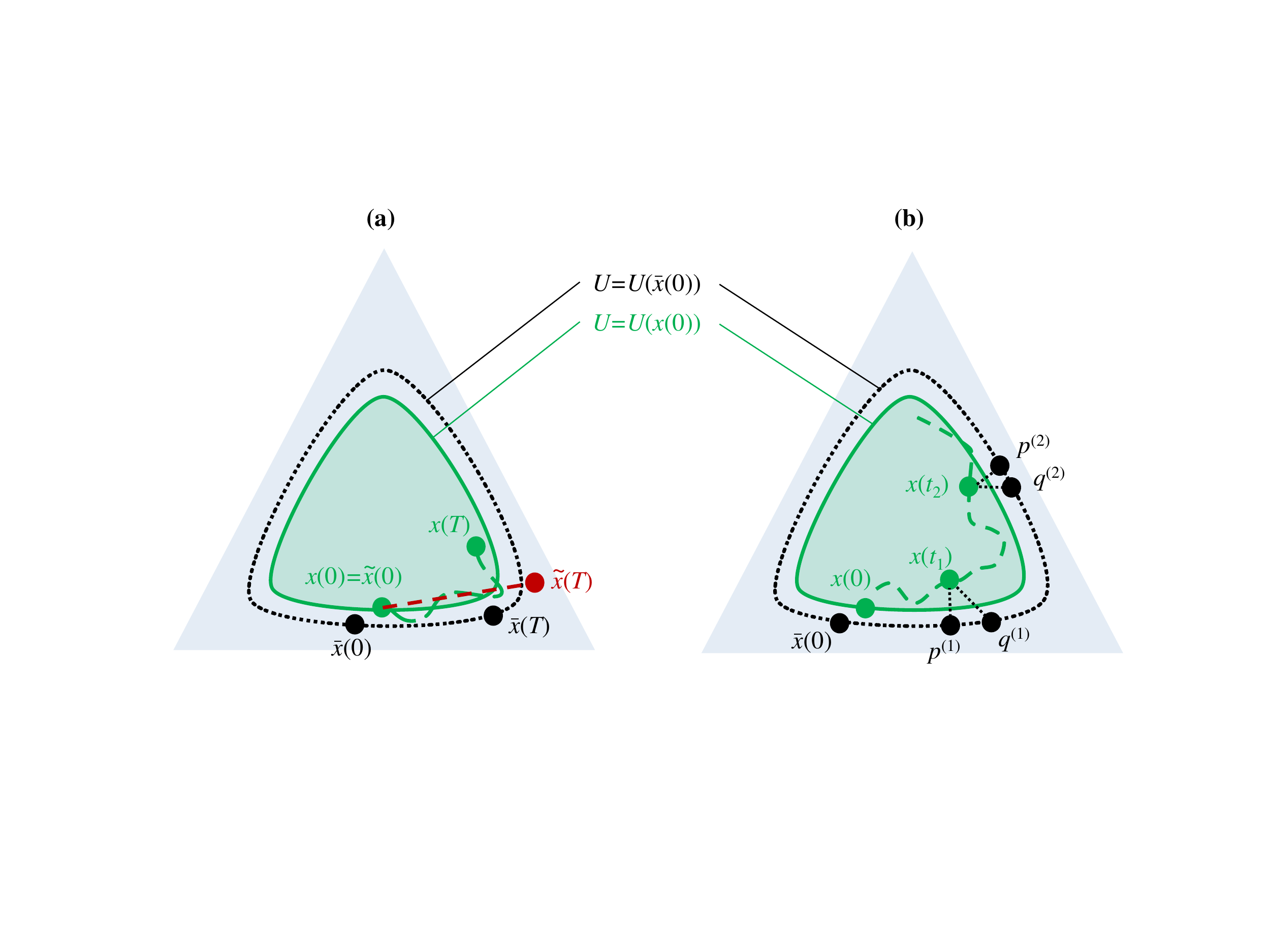}
\fi
\caption{Illustration of the proof of Lemma~\ref{lem:infdist}: (a) Short-term view ($T < n^{2/3}$) and (b) Long-term view ($T>n^{2/3}$).}
\label{fig1}
\end{figure}
To show Claim (i), we apply the concentration bound~\eqref{eq:conc} to each of the random processes $X_a(t) = x_a(t) - \contz x_a(t)$, for $a\in \{1,2,3\}$. We have $X_a(t) = \sum_{\tau=1}^t \delta_a(\tau)$, where $\delta_a(t)$ is given as:
$$
\delta_a(t) = \Delta x_a(t-1) - \frac{1}{n} x_a(0)( x_{a+1}(0) -  x_{a+2}(0)),
$$
and $\Delta x_a$ is a random variable with distribution described by \eqref{eq:deltana}. We have $|\delta_a(t)| \leq \frac{2}{n}$ and
\begin{equation}\label{eq:edelta}
|\E\delta_a(t)| = \frac{1}{n} |x_a(t-1)( x_{a+1}(t-1) -  x_{a+2}(t-1)) - x_a(0)( x_{a+1}(0) -  x_{a+2}(0))| < \frac{6t}{n^2} \leq \frac{6T}{n^2},
\end{equation}
where we took into account that $|x_i(t-1)-x_i(0)|<t/n$, for all $i\in\{1,2,3\}$ by the definition of the population dynamics, which changes the size of each species by at most $1$ in each step. Thus, applying the concentration bound with $c = \frac{2}{n}$, $\alpha = \frac{6T}{n^2}$, and $\eps = \log^2 n$, we have that for $T\leq n^{2/3}$:
$$
\forall_{t\in\{1,\ldots,T\}}\
|x_a(t) - \contz x_a(t)| = |X_a(t)| \leq  \frac{2 T^{0.5}\log n}{n} + \frac{6 T^2}{n^2} = \Otilde(T^{0.5}/n)
$$
holds with very high probability $1-e^{-\Omega(T^{\eps})} = 1-e^{-\Omega(\log^2 n)}$. Applying the union bound over all $a\in \{1,2,3\}$, we obtain the claim.

We now proceed to prove Claim (ii), assuming $T \leq n^{2/3}$ (see Fig.~\ref{fig1}(a)). By Claim (i), we have $d_\infty(x(t),\contz x(t)) = \Otilde(T^{0.5}/n)$ for all $t \leq T$, w.v.h.p. It remains to observe that the continuous process is also close to the linear approximation $\contz x$ (compare with~\eqref{eq:edelta}):
\begin{align*}
&d_\infty(\cont x(t),\contz x(t)) \leq \\
&\leq \Delta + \max_{a\in\{1,2,3\}}\int_{\tau=0}^t \frac{1}{n}\left| \cont x_a(\tau)( \cont x_{a+1}(\tau) - \cont  x_{a+2}(\tau)) - x_a(0)( x_{a+1}(0) -  x_{a+2}(0))\right| d\tau \leq\\
&\leq \Delta + \max_{a\in\{1,2,3\}}\int_{\tau=0}^t \frac{1}{n}\left| \cont x_a(\tau)( \cont x_{a+1}(\tau) - \cont  x_{a+2}(\tau)) - \cont x_a(0)( \cont x_{a+1}(0) -  \cont x_{a+2}(0))\right| d\tau +\\
& \phantom{\leq}+ \max_{a\in\{1,2,3\}}\int_{\tau=0}^t \frac{1}{n}\left| \cont x_a(0)( \cont x_{a+1}(0) - \cont  x_{a+2}(0)) - x_a(0)( x_{a+1}(0) -  x_{a+2}(0))\right| d\tau \leq\\
&\leq \Delta + \int_{\tau=0}^t \frac{6\tau}{n^2} d\tau
+ \int_{\tau=0}^t \frac{6\Delta}{n} d\tau
= \Delta + \frac{3t^2}{n^2} + \frac{6 t \Delta}{n}\leq \Delta + \frac{3T^2}{n^2} + \frac{6T \Delta}{n} = \Oh(\Delta + T^{0.5}/n).
\end{align*}
Taking into account that $d_\infty(x(t),\cont x(t)) \leq  d_\infty(x(t),\contz x(t)) + d_\infty(\contz x(t),\cont x(t))$, we obtain the claim.

To prove Claim (iii), we only need to consider the case of $T> n^{2/3}$, which is not covered by the stronger Claim (ii). Let $0=t_0 < t_1 < \ldots < t_l = T$ be fixed integer moments of time chosen so that $t_i - t_{i-1} \leq n^{2/3}$, for $1\leq i \leq l$, and $l = \Oh(T / n^{2/3})$. For $1\leq i \leq l$, let $p^{(i)}$ denote the random variable representing the point on the orbit of the considered continuous dynamics $\cont x$, located closest in the infinity norm to the point $x(t_i)$ obtained after $t_i$ steps of evolution of the discrete dynamics:
$$
p^{(i)} \in \arg\min_p \{d_\infty(p, x(t_i)) : U(p) = U(\cont x(0))\},
$$
and for consistency of notation, let $p^{(0)} = \cont x(0)$.
By Lemma~\ref{lem:convU}, we have w.v.h.p.\ that for all $0\leq i \leq l$:
$
d_U (x(0), x(t_i)) = \Otilde(T^{0.5}/n^{1-\gamma})
$
and taking into account that
$
d_U(x(0), p^{(i)}) \leq \Delta,
$
we obtain
$$
d_U (x(t_i), p^{(i)}) = \Otilde(\Delta + T^{0.5}/n^{1-\gamma}).
$$
By Lemma~\ref{lem:duinfty}, we have $d_\infty (x(t_i), p^{(i)}) < d_U (x(t_i), p^{(i)})$, thus,
\begin{equation}\label{eq:distxp}
d_\infty (x(t_i), p^{(i)})= \Otilde(\Delta + T^{0.5}/n^{1-\gamma}).
\end{equation}

Now, for $1\leq i \leq l$, define $q^{(i)}$ as the point obtained by applying the continuous dynamics~\eqref{eq:contrps} to point $p^{(i-1)}$ for $\delta t_i = t_i - t_{i-1}$ steps. Note that $q^{(1)} = \cont x(t_1)$, but that for $i>1$, a similar correspondence need not hold precisely. However, since $\delta t_i \leq n^{2/3}$, we can apply Claim $(ii)$ to the discrete dynamics starting at point $x(t_{i-1})$ and the continuous dynamics starting at point $p^{(i-1)}$ and running for $\delta t_i$ steps. Substituting $\Otilde(\Delta + T^{2/3}/n^{1-\gamma})$ for ``$\Delta$'' and $\Oh(n^{2/3})$ for ``$T$'' in Claim (ii), we have:
\begin{equation}\label{eq:distxq}
d_\infty (x(t_i), q^{(i)}) = \Otilde\left(\frac{(n^{2/3})^{0.5}}{n} + (\Delta + T^{0.5}/n^{1-\gamma})\right) = \Otilde (\Delta + T^{0.5}/ n^{1-\gamma}),
\end{equation}
with very high probability. Combining the above bound with~\eqref{eq:distxp}, we obtain:
$$
d_\infty (p^{(i)}, q^{(i)}) = \Otilde (\Delta + T^{0.5}/ n^{1-\gamma}).
$$
By applying the union bound, the above also holds for all $1 \leq i \leq l$, w.v.h.p. Now, we introduce the following auxiliary notation. For any pair of points $p, q$, such that $U(p) = U(q)$, we define by $\cont t(p,q)$ the smallest value of time $t$ such that either point $p$ is obtained by applying the continuous dynamics \eqref{eq:contrps} to starting point $q$ for time $t$, or vice versa. For the considered set of points, the following bound holds:
$$
\cont t(\cont x(T), q^{(l)}) \leq \sum_{i=1}^{l-1} \cont t (p^{(i)}, q^{(i)}).
$$
By Lemma~\ref{lem:disttime}, we obtain that for a pair of points $p,q$ at distance $o(n)$ on the considered orbit $-10 > U = U(\cont x (0)) > -\gamma\ln n$, the following relations hold:
$$
d_\infty (p,q) = \Omega(\cont t (p,q)/n^{1+\gamma})
$$
and
$$
d_\infty (p,q) = \Oh(\cont t (p,q)/n).
$$
Thus, we have:
\begin{align*}
d_\infty (\cont x(T), q^{(l)}) &= \Otilde( l \cdot (\Delta + T^{0.5} n^{\gamma-1})\frac{n^{1+\gamma}}{n}) = \Otilde( (T n^{\gamma-2/3}) \cdot (\Delta + T^{0.5} n^{\gamma-1})).
\end{align*}
Combining the above with bound~\eqref{eq:distxq} for point $q^{(l)}$, we get:
\begin{align*}
d_\infty (\cont x(T), x(T)) &= \Otilde( (T n^{\gamma-2/3}) \cdot (\Delta + T^{0.5} n^{\gamma-1})) + \Otilde (\Delta + T^{0.5} n^{\gamma-1})=\\
&= \Otilde( (T n^{\gamma-2/3}) \cdot (\Delta + T^{0.5} n^{\gamma-1}))
\end{align*}
The above holds w.v.h.p. By stating the obtained claim for all values of $T' \in [n^{2/3}, T]$, and applying a union bound over all such $T'$, as well as for the case of smaller $T'$ covered by Claim (ii), we obtain Claim (iii).

To show claim (iv), we consider the set of points $p(U)$ which lie at minimum distance $d_\infty$ from point $\cont x (T)$, taken over all potential orbits $U = U(x(0))\pm \Otilde (T^{0.5}/n^{1-\gamma})$. Taking into account Lemma~\ref{lem:convU}, we will show that the evolution $x(t)$ will intersect with set $p(U)$ after $T' = (1\pm o(1)) T$ steps; in this way, we will obtain the claim directly from Lemma~\ref{lem:duinfty}. By claim (iii), after $T$ steps, our discrete evolution has reached a point $x(T)$ such that $d = d_\infty(x(T),\cont x(T)) = \Otilde( (T n^{\gamma-2/3}+1) \cdot (\Delta + T^{0.5} n^{\gamma-1})) \leq \Otilde( T n^{-1-2\gamma})$, w.v.h.p. Taking into account Lemma~\ref{lem:disttime}, we conclude that the intersection of the trajectory with set $p(U)$ took place at time $T' = T \pm O(d n^{1+\gamma}) = (1\pm o(1)) T$.
\end{proof}
}

\InJournal{
\subsection{The coupling framework}\label{sec:rpscoupling}

For a fixed absorbing state $r \in \{(1,0,0), (0,1,0), (0,0,1)\}$, we denote by $p_r(x)$ the probability that the RPS population protocol~
\eqref{eq:deltana}, starting from initial state $x$, reaches absorbing state $r$. To show that for a pair of states $x,y$ we have $p_r(x)
\approx p_r(y)$, we define a coupling on the evolution of states $x$ and $y$ over time. Formally, we denote by $R$ the set of possible actions which can be taken in any step by the considered population protocol. In the context of RPS, we define $R = \{1,\ldots, n(n-1)\}$, and treat elements of $R$ as identifiers of ordered pairs of elements of the population. By convention, we will assume that the elements are ordered in each step, so that in a state $(x_1, x_2, x_3)$, the $x_1$ elements belonging to species $1$ have the smallest identifiers, while the $x_3$ elements belonging to species $3$ have the largest. For any $s\geq 0$ and sequence $S = (S_1, \ldots, S_s) \in R^s$, we denote by $x[S]$ the state reached by executing the population protocol for $s$ steps, starting from state $x$, choosing in the $i$-th step, $1\leq i \leq s$, an action of the protocol encoded by the value $S_i \in R$.

We will define $F$ to be a random number generator over $R$ if subsequent calls to $F$ return elements from $R$ chosen uniformly and independently at random; formally, $F = (F_1, F_2, \ldots,)$, where all $F_i \sim unif(R)$ are independent random variables. By a slight abuse of notation, we will denote by $x[F,s]$ the random variable in the state space corresponding to the sequence of actions $x[(F_1, \ldots, F_s)]$.

We now provide a method for ``skipping'' some elements returned by a random number generator, to create a new random number generator. We recall that if the decision whether or not to skip the $i$-th element in the sequence returned by a random number generator depends only on the values of the previously generated elements (up to the $(i-1)$-st, then such a generator remains unbiased. Formally, we call a sequence of functions $\phi = (\phi_1, \phi^\circ_1, \phi_2, \phi^\circ_2, \ldots)$, with $\phi_i : R^{i-1} \to \{0,1\}$ being measurable (an elimination function) and $\phi^\circ_i : F_i \mapsto \phi^\circ_i(F_i) \sim unif(R)$ acting as a transformation of random variable $F_i$, a \emph{resampler} on random number generators. For a random number generator $F$, we denote by $\phi F = ((\phi F)_1, (\phi F)_2, \ldots)$ the sequence of random variables with values in $R$, realized inductively as follows. The first random variable $(\phi F)_1$ takes the value of $\phi^\circ_{i_1}(F_{i_1})$, where $i_1 > 0 $ is the smallest index such that $\phi_{i_1} (F_1, F_2, \ldots, F_{i_1-1}) = 1$. Next, $(\phi F)_2$ takes the value of $\phi^\circ_{i_2}(F_{i_2})$, where $i_2 > i_1$ is the smallest index larger than $i_1$ such that $\phi_{i_1} (F_1, F_2, \ldots, F_{i_2-1}) = 1$, and so on. We note that the resampled distribution $\phi F$ is also a (uniform, independent) random number generator over $R$.

We are now ready to formulate the standard coupling technique in terms of applying resamplers to random number generators.

\begin{proposition}[Coupling of delayed walks]
Let $F$ be a random number generator over $R$, let $x, y$ be two arbitrary starting points in the state space, and let $r$ be an arbitrary absorbing state. If there exist two filters $\phi_x$, $\phi_y$ on a random number generator such that $\lim_{t\to +\infty} (x[\phi_x F,t]) = \lim_{t\to +\infty} (y[\phi_y F,t])$ holds with probability $1-\eps$, for some $\eps \geq 0$, then $|p_r(x) - p_r(y)|\leq \eps$.
\end{proposition}

The validity of the above proposition follows from the fact that each of the marginals $\phi_x F$, $\phi_y F$ represents an unbiased evolution of the system. Informally, we will apply the above technique as follows. We will sample numbers from $R$ from a random number generator. Based on all the numbers sampled so far, we will decide if the \emph{next} number to be sampled should be applied only in the evolution of the process originating from $x$, in the evolution of the process originating from $y$, or in both those processes. (If only one process undergoes evolution in the given step, the other process can be thought of as delayed in this step). In the following, we will consider the walks in the state space originating from $x$ and $y$, and we will execute the following types of phases:
\begin{itemize}
\item One of the walks progresses until a given termination condition is met, while the other walk is delayed.
\item Both walks progress simultaneously in a given step $t$, with one walk following action $\phi_{x\ t}^{\circ}(F_t)$ and the other following action $\phi_{y\ t}^{\circ}(F_t)$ in this step.
\end{itemize}
}
\noindent
We are now ready to apply the coupling technique to obtain the main technical result of this section.

\begin{lemma}\label{lem:coupling}
Fix $\gamma = 0.005$ and $\eps = 0.05$. Let $x(0) = (x_1(0), x_2(0), x_3(0))$ be arbitrarily fixed with $-12 > U(x(0)) > -\gamma \ln n$, and let $y(0) = (x_3(0), x_1(0), x_2(0))$. Then, there exist a coupling of $x$ and $y$ which leads to the same absorbing state with probability $1 - \Otilde (n^{-\eps})$.
\end{lemma}
\InConference{\begin{proof}[sketch]}
\InJournal{\begin{proof}}
The proof proceeds by a coupling of walks originating from $x$ and $y$\InJournal{ in the above-described framework}. By a slight abuse of notation, we will denote by $x(t)$ and $y(t)$ the position of each of the two walks in the state space after $t$ steps\InJournal{ (i.e., after $t$ applications of the random generator $F$)}, which may include steps in which a given walk is delayed. Our goal is to make points $x(t)$ and $y(t)$ coalesce within a small number of steps $T$, i.e., to obtain $x(T)=y(T)$ with probability $1-\Oh(n^{-\eps})$, where $T \ll n^{1.33}$. \InJournal{Consequently, taking into account Lemma~\ref{lem:convU}, we will silently assume that the bounds $-(1+o(1)) 12 > U(x(t)) > - (1+o(1))\gamma \ln n $ and $-(1+o(1)) 12 > U(y(t)) > - (1+o(1))\gamma \ln n $ are preserved throughout the coupling.}

The coupling proceeds in five phases\InJournal{, illustrated in Fig.~\ref{fig2}}. \InJournal{We start by providing a high level overview,}\InConference{We limit this proof sketch to a high-level overview,} thinking for now of the potential $U(x(0)) = \Theta(1)$ to simplify calculations. In Phase 1, point $x$ approaches point $y$, which is stopped. In this way, the infinity norm distance between $x$ and $y$ is reduced, at the cost of increasing the $d_U$ distance to slightly over $n^{-0.5}$. Next, in Phase 2 we run both walks independently, so that the distance $d_U$ in time follows a random evolution resembling a random walk, and after slightly more than $n$ steps, the value $d_U = 0$ is hit with sufficiently high probability. Whereas the walks are now orthogonally aligned, we also need to align them along the orbit, since we may at this point have them at a distance of $d_\infty > n^{-1/3}$ apart. By allowing the slower walk to catch up, we reduce $d_\infty$ to slightly more than $n^{-2/3}$, at the cost of increasing $d_U$ to a similar value. In this way, we have decreased the norm in both distances (from about $n^{-0.5}$ to about $n^{-2/3}$). We iterate Phase 2, reducing each time the distance between the two walks in both norms, up to an iteration in which the size of the populations in $x$ and $y$ differ by an arbitrarily small polynomial in $n$. \InJournal{(We remark that multiple iterations of Phase 2 may also be replaced by a path-coupling argument of the type used in~\cite{LPW06}, but this does not necessarily simplify the proof.)} At this point, only a very small number of time steps remains until coalescence. We first align the two states by evolving one of them until the size of one of the three populations is identical for $x$ and $y$ at the end of Phase 3, and then perform a standard coupling by correlating the evolution of $x$ and $y$ in Phase 4, so as to make the sizes of the other two populations meet for $x$ and $y$, while maintaining equality on the size of the population coalesced in Phase 3. Finally, after the coupling is achieved, we evolve the coalesced state into an absorbing state in Phase~5.
\InJournal{

\begin{figure}[ht]
\centering
\ifpdf
\mbox{%
\framebox{\includegraphics[height=6.5cm,trim=4.5cm 0.5cm 2.5cm 0.0cm,clip]{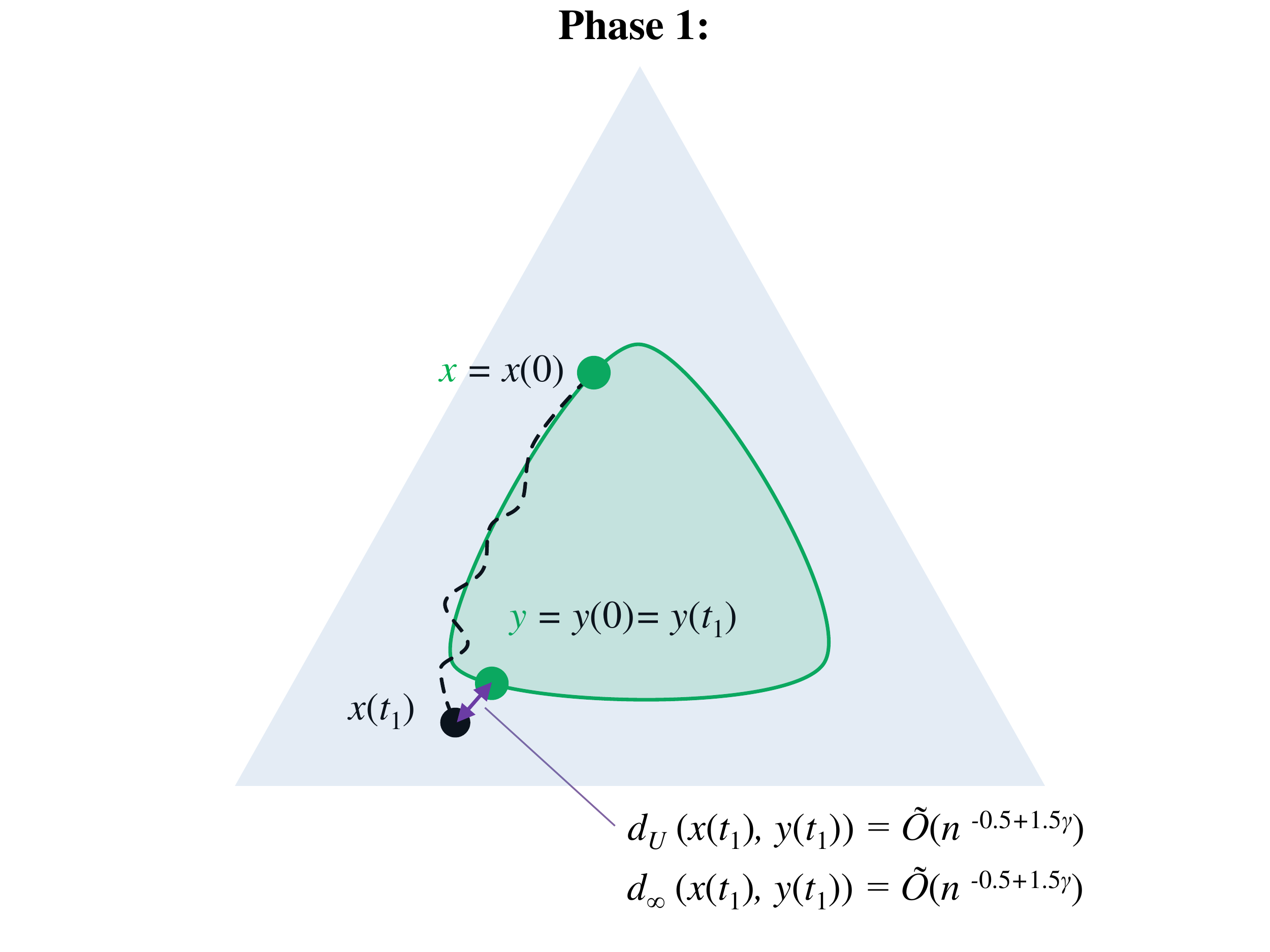}}
\framebox{\includegraphics[height=6.5cm,trim=0.5cm 0.5cm 4.5cm 0.0cm,clip]{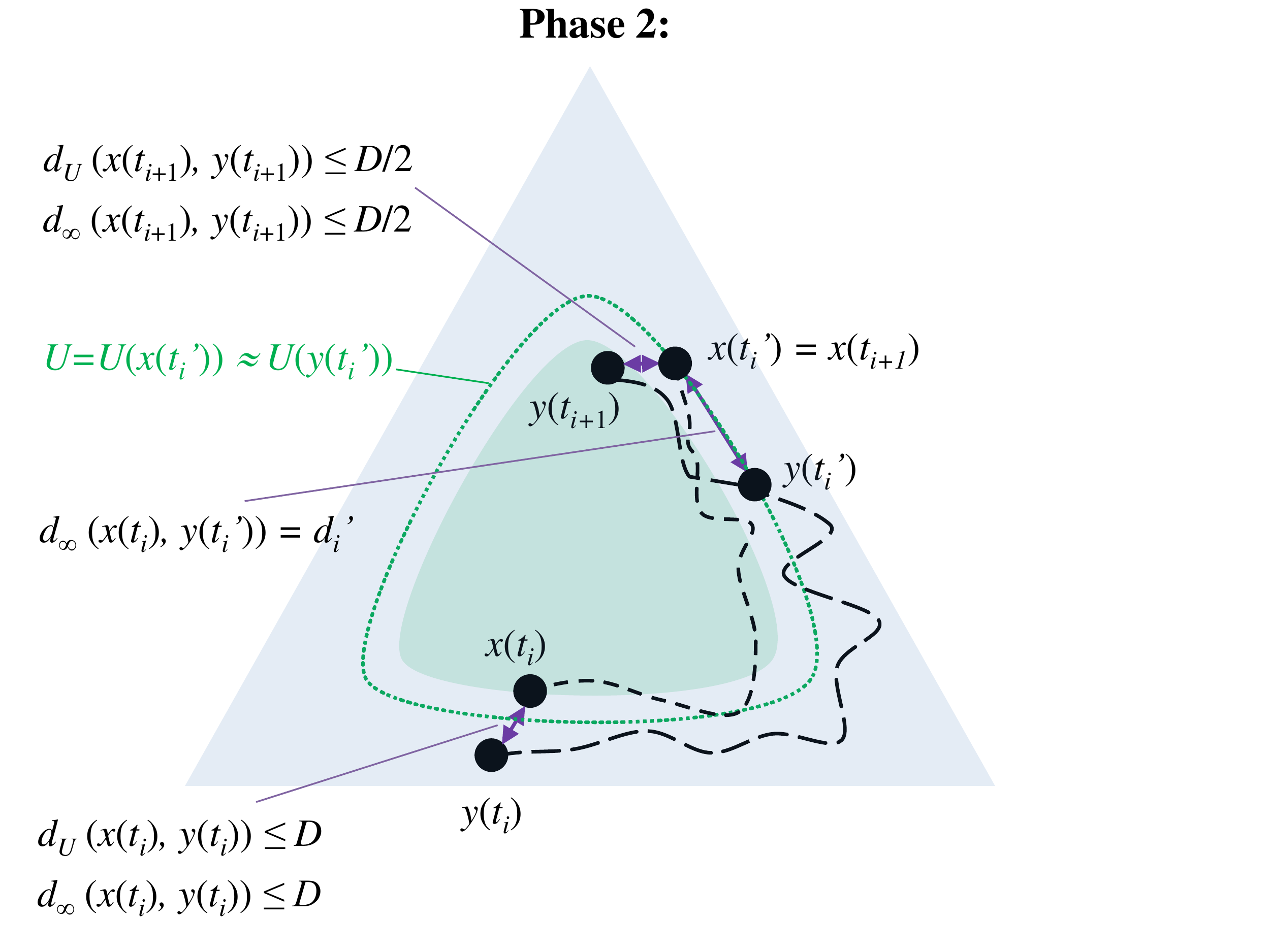}}}\\[0.8mm]
\framebox{\hspace{1.035cm}\includegraphics[height=3cm,trim=0cm 2.5cm 0cm 8cm,clip]{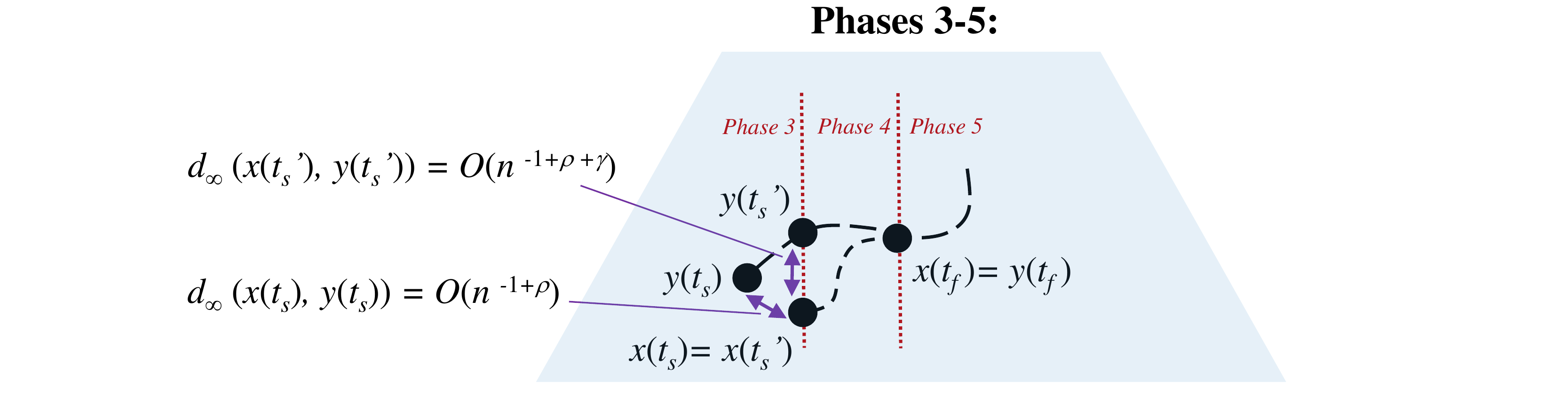}\hspace{1.035cm}}
\fi
\caption{Illustration of the coupling for the proof of Lemma~\ref{lem:coupling}}\label{fig2}
\end{figure}
\begin{itemize}
\item \emph{Phase $1$}. In the first $t_1$ steps, for all $t < t_1$ the walk $x(t)$ progresses while the walk $y(t)$ is delayed. We choose the duration of the phase so that $x(t_1)$ and $y(t_1)$ as close to each other in the infinity norm as possible, after walk $x(t)$ has traversed approximately $1/3$ of its orbit around the state space. Formally, let $T$ be the number of steps of the continuous evolution which transforms point $x(0)$ into point $y(0)$; we have by Lemma~\ref{lem:disttime} that $t_1 = \Oh(n / e^{U(x)}) = \Oh(n^{1+\gamma})$. Next, let $t_1$ be chosen as the time $T'=(1+o(1))T$ which follows from Lemma~\ref{lem:infdist}(iv), i.e., the time such that the discrete process starting from $x(0)$ after $T'$ steps satisfies $d_\infty(y(0),x(T')) = \Otilde(T^{0.5}n^{\gamma-1}) = \Otilde(n^{-0.5+1.5\gamma})$), w.v.h.p. Moreover, we also have $d_U(y(0),x(T')) = \Otilde(T^{0.5}n^{\gamma-1}) = \Otilde(n^{-0.5+1.5\gamma})$) by Lemma~\ref{lem:convU}, w.v.h.p. Since walk $x$ progresses for $t_1 = T'$ steps and walk $y$ is delayed throughout ($y(t_1) = y(0)$), at the end of this phase, we obtain $d_U(y(t_1), x(t_1)) = \Otilde(n^{-0.5+1.5\gamma})$ and $d_\infty(y(t_1), x(t_1)) = \Otilde(n^{-0.5+1.5\gamma})$, w.v.h.p.

\item \emph{Phase $2$} consists of a certain number (logarithmic in $n$) of iterations. We start iteration $i$ at time $t_i$ with $d_U(y(t_i), x(t_i)) \leq D$ and $d_\infty(y(t_1), x(t_1)) \leq D$, where $D\leq n^{-0.5+1.5\gamma}$. In each iteration, we start by running both walks in parallel until a time $t'_i$ at which the two evolutions are located on the same potential orbit (up to integer rounding): $d_U(y(t'_i), x(t'_i)) = \Oh(n^{\gamma-1})$. We introduce the following notation: let $r(t) \equiv U(x(t_i+t)) - U(y(t_i + t))$ and consider the process $X(t) = r(t) - r(0)$. We can use the following sum representation: $X(t) = \sum_{\tau=1}^t \delta(\tau)$, with the increment $\delta(\tau)$ given as $\delta(\tau) = r(\tau) - r(\tau-1).$ Notice that $|r(0)| = d_U(y(t_i), x(t_i)) = \Oh(n^{\alpha_i})$. We now apply the anticoncentration bound of Lemma~\ref{lem:martingale} to process $X(t)$ with $|r(0)|= d_\infty(y(t_i), x(t_i)) \leq D$ to bound the probability of $r(t)$ having opposite sign to $r(0)$, for some $t \in [0,T]$, where the length of the interval is suitably chosen ($T = n^{2 + 7 \gamma + \eps} D^2$). Without loss of generality of the argument, let $r(0)\leq 0$; then, we  and ask about the probability of the event $r(t)\geq 0$, which is equivalent to $X(t) = r(t) - r(0) > -r(0) = D$ occurring. (If $r(0)>0$, we apply the anticoncentration bound to process $-X$ instead.) Now, by analyzing the change of potential $U$ in a step of the discrete dynamics~\eqref{eq:deltana}, for the different pairs of possible agents, we obtain that Lemma~\ref{lem:martingale} is applicable to process $X$ with the following parameters (compare with the proof of Lemma~\ref{lem:convU} for a related derivation):
$$
\begin{cases}
c = \Otilde(n^{-1+\gamma})\\
\alpha = \Otilde(n^{-2+\gamma})\\
\sigma = \Omegatilde(n^{-1-0.5\gamma})\\
T = n^{2 + 7 \gamma + \eps} D^2,\\
\end{cases}
$$
where $\eps$ is an absolute constant belonging to the range $0 < \eps < 0.25 - 4.75\gamma$.

Then, we can bound the probability of failure for equation~\eqref{eq:anticonc} as follows:
\begin{align*}
&82 \left(\frac{c^2}{\sigma^2}\frac{(D+\alpha T) \ln T}{\sigma \sqrt T}\right)^{2/3} =
\Otilde\left(\left(\frac{(n^{-1+\gamma})^2}{(n^{-1-0.5\gamma})^2}\frac{D+ n^{\gamma-2} T}{n^{-1-0.5\gamma} \sqrt T}\right)^{2/3}\right) =\\
&=\Otilde\left(n^{2/3+7\gamma/3}\frac{D^{2/3}}{T^{1/3}} + n^{-2/3+3\gamma} T^{1/3}\right)=\Otilde(n^{-\eps}).
\end{align*}
At this point, the time shift of the two processes can be bounded as:
\begin{align*}
d'_i = d_{\infty}(x(t'_i), y(t'_i)) &=
\Otilde( (T n^{\gamma-2/3}+1) \cdot (D + T^{0.5} n^{\gamma-1})) = \\
& = \Otilde((D^2 n^{4/3+ 8\gamma + \eps}+1)D n^{4.5\gamma + 0.5\eps}).
\end{align*}
Now, we proceed from time $t'_i$ to time $t_{i+1}$ by aligning the two walks. Without loss of generality, suppose that $x(t'_i)$ is lagging behind along its trajectory with respect to $y(t'_i)$ (i.e., that a continuous evolution from point $x(t'_i)$ reduces its distance to point $y(t'_i)$). We now continue the evolution of point $x(t'_i)$ until it becomes close to point $y(t'_i)$, which remains motionless. By Lemma~\ref{lem:infdist} (iv), we then have:
\begin{align*}
d_{\infty}(x(t_{i+1}), y(t_{i+1})) &= \Otilde \left(\frac{(t_{i+1}-t'_i)^{0.5}}{n^{1-\gamma}}\right) = \Otilde \left(\frac{({d'_i} n^{1+\gamma})^{0.5}}{n^{1-\gamma}}\right)=\\
& = \Otilde((D n^{2/3+ 4\gamma + 0.5\eps}+1)D^{0.5} n^{-0.5+3.75\gamma+0.25\epsilon}).
\end{align*}
For $D > n^{-2/3 -4\gamma-0.5\eps}$, we further bound this as:
\begin{align*}
d_{\infty}(x(t_{i+1}), y(t_{i+1})) &= \Otilde(D n^{2/3+ 4\gamma + 0.5\eps} (n^{-0.5+1.5\gamma})^{0.5} n^{-0.5+3.75\gamma+0.25\epsilon}) = \\
&= \Otilde(D n^{-1/12+ 8.5\gamma + 0.75\eps}) < D/2 =  d_{\infty}(x(t_{i}), y(t_{i})) /2,
\end{align*}
where the final inequality holds when we set $\gamma = 0.005$, $\eps = 0.05$, and sufficiently large $n$.

For $D \leq n^{-2/3-4\gamma-0.5\eps}$, we apply the following bound:
\begin{align*}
d_{\infty}(x(t_{i+1}), y(t_{i+1})) &= \Otilde(D^{0.5} n^{-0.5+3.75\gamma+0.25\epsilon})  < D/2 =  d_{\infty}(x(t_{i}), y(t_{i})) /2,
\end{align*}
where the final inequality holds for sufficiently large $n$ when $D = \omega(n^{-1+7.5\gamma+0.5\epsilon})$.

We also recall that by Lemma~\ref{lem:infdist} (iv), we can bound $d_{U}(x(t_{i+1}), y(t_{i+1}))$ similarly to $d_{\infty}(x(t_{i+1}), y(t_{i+1}))$, obtaining $d_{U}(x(t_{i+1}), y(t_{i+1})) < D/2$ under the same assumptions.

Thus, in every iteration of the coupling phase, starting from a pair of points $(x(t_i), y(t_i))$ such that $D = \omega(n^{-1+7.5\gamma+0.5\epsilon})$, with probability $1- \Otilde(n^{-\eps})$ we reach in $O(n^{2 + 7 \gamma + \eps} D^2)$ steps a new pair of points $(x(t_{i+1}), y(t_{i+1}))$ such that $d_{\infty}(x(t_{i+1}), y(t_{i+1}))<D/2$ and $d_{U}(x(t_{i+1}), y(t_{i+1}))<D/2$. The next iteration then starts. The phase is completed after less than $s$ iterations, where $s = \Oh(\log n)$, with success probability $1- \Otilde(n^{-\eps})$. For the chosen values of constants $\gamma$ and $\eps$, at the end of the phase, we obtain that $D < n^{-1+\rho}$, where $\rho = 0.07$.

\item \emph{Phase 3.} We start this phase from a pair of points $(x(t_{s}), y(t_{s}))$, such that $d_\infty(x(t_{s}), y(t_{s})) < n^{-1+\rho}$. Our goal is to reach in a small number of steps a pair of points $(x(t'_{s}), y(t'_{s}))$ such that one of the species has populations of exactly the same size in $x(t'_{s})$ and $y(t'_{s}))$, without increasing $d_\infty(x(t'_{s}), y(t'_{s}))$ too much. Let $a$ be the species such that the value of $x_{a+1}(t_{s}) - x_{a+2}(t_{s})$ is maximized over $a \in\{1,2,3\}$ ($a$ represents the species which is most likely to grow quickly in the next few steps). We then continue to evolve that of the walks $x(t_{s}), y(t_{s})$, which has a smaller size of population $a$, while the other walk remains paused throughout this phase. The phase ends at a time $t'_{s}$, such that $x_a(t'_s) = y_a(t'_s)$. We restrict ourselves to a consideration of the case when $y_a(t_s) < x_a(t_s)$ (walk $y$ progresses, while $x$ is paused). Then, at each of the considered steps $t_s + t$, for $0 \leq t < n^{0.5}$, we have:
    \begin{align*}
    &\E(y_a(t_s+t+1) - y_a(t_s+t)) = \frac{y_a(t_s+t)}{n}(y_{a+1}(t_s+t) - y_{a+2}(t_s+t))\geq \\
    &\geq \frac{y_a(t_s)}{n}(y_{a+1}(t_s) - y_{a+2}(t_s))-\frac{3t}{n^2} \geq n^{-1-\gamma}(x_{a+1}(t_s) - x_{a+2}(t_s) - 2n^{-1+\rho})-\frac{3t}{n^2} =\\ &= \Omega(n^{-1-\gamma}).
    \end{align*}
    The change of $y_a$ in each step of the walk is bounded by $1/n$. Applying Azuma's inequality, we obtain that w.v.h.p. the deviation of $y_a$ from its expectation is bounded by $\Otilde(\sqrt t/n)$. Thus, w.v.h.p., for $t = c n^{\gamma+\rho}$ for a sufficiently large constant c, we have:
    $$
    y_a(t_s + t) = y_a(t_s) + \Omega(n^{-1-\gamma} c n^{\gamma+\rho} \geq y_a(t_s) + d_\infty(x(t_{s}), y(t_{s})) > x_a(t_s).
    $$
    Thus, the phase is completed within $O(n^{\gamma+\rho})$ steps w.v.h.p. At the end of the phase, we have $x_a(t'_s) = y_a(t'_s)$, and $d_\infty(x(t'_{s}), y(t'_{s})) = O(n^{-1 + \gamma+\rho})$.

\item \emph{Phase 4.}
In this penultimate phase, we allow both walks to progress simultaneously in each time step starting from step $t'_s$ until they coalesce, under the following coupling. At time step $t'_s$, we arrange the elements of the populations of $x(t)$ and $y(t)$ in linear order, from $1$ to $n$, putting those from population $1$ leftmost and those from population $3$ rightmost. Every step of a coupling corresponds to drawing a uniformly random pair $(i,j) \in \{1,\ldots,n\}^2$ from our random number generator, which corresponds to an interaction in which the $i$-th and $j$-th agents in $x(t)$ and $y(t)$ interact with each other (the roles of $i$ and $j$ as predator and prey will be defined later). We label the species of the agents as follows $s : \{1,\ldots,n\} \to \{-1, 1, 2, 3\}$, where $s(i) = 1, 2,$ or $3$ if both of the agents located at the $i$-th position (counting from the left) in $x(t'_s)$ and $y(t'_s)$ belong to the same respective species, and $s(i) = -1$ otherwise. We remark that by the condition $d_\infty(x(t_{s}), y(t_{s})) = n^{-1 + \gamma+\rho}$, we have that $|\{i : s(i) = -1\}| = \Oh (n^{-1 + \gamma+\rho})$. Note that the labels $s$ are defined at time $t'_s$ and never change afterwards.

The coupling process will be required to finish within $T= n^{0.5-\eps}$ time steps of Phase 4. We will only consider the success of the coupling under the event that the generator never makes any agent participate in more than one interaction throughout Phase 4. (By a birthday paradox computation, this event holds with probability $1 - \Otilde(n^{-\eps})$.) We will also condition the success of our coupling on the event that no agent $i$ with $s(i)=-1$ is chosen to interact throughout the phase (this event holds with probability $1 - \Oh (T n^{-1 + \gamma+\rho}) > 1 - \Oh (n^{-\eps})$). Under these conditions, we define the interactions for a drawn random pair $(i,j)$ in step $t$ of Phase 4 as follows (independently of the step number, until the termination condition $x(t) = y(t)$ is reached):
\begin{itemize}
\item If $\{s(i), s(j)\} = \{a+1,a+2\}$, then we perform the interaction $i\to j$ with probability $0.5$ and the interaction $j\to i$ with probability $0.5$, chosen independently for processes $x(t)$ and $y(t)$. We recall that species $a$ was defined in Phase 3 so that $x_a(t'_s) = y_a(t'_s)$.
\item If $\{s(i), s(j)\} \neq \{a+1,a+2\}$, then we perform the interaction $i\to j$ for both processes.
\end{itemize}
Notice that since we take into account only the interaction of agents having a positive value of $s(i)$, and we never allow any agent to interact more than once, all interactions performed by agents from population $a$ will have the same outcome for both processes $x$ and $y$. Thus, we have $x_a(t) = y_a(t)$ throughout Phase 4, and it suffices to show that a similar equality will hold for the other populations. Equivalently, we will require that $Z(t) = X(t) - Y(t) = 0$, where the processes $X$ and $Y$ are defined as $X = x_{a+1}(t) - x_{a+2}(t)$ and $Y = y_{a+1}(t) - y_{a+2}(t)$. Let $p = \frac{|\{i : s(i) = a+1\}|\cdot |\{i : s(i) = a+2\}|}{n^2} = \Omega(n^{-2\gamma})$. Observe that $X$ and $Y$ are walk processes along the real line, such that in each step we have $\Pr [X(t+1) - X(t) = 1] = p$ and $\Pr [X(t+1) - X(t) = 0] = 1-p$. Likewise, we have independently $\Pr [Y(t+1) - Y(t) = 1] = p$ and $\Pr [Y(t+1) - Y(t) = 0] = 1-p$. (Note that probability $p$ does not change over time.) Thus, we have that process $Z$ follows a lazy random walk on the integer axis, with $\Pr [Z(t+1) - Z(t) = 1] = \Pr [Z(t+1) - Z(t) = -1] = 2p(1-p)$ and $\Pr [Z(t+1) - Z(t) = 0] = 1 - 2p(1-p)$. Since $Z(t'_s) = \Oh (n d_\infty(x(t'_{s}), y(t'_{s})) = n^{\gamma+\rho})$, by the standard properties of hitting times of the random walk, we have that $Z(t'_s + t) = 0$ is achieved within time $t = \Otilde (\frac{1}{2p(1-p)} (n^{\gamma+\rho})^2) = \Otilde (n^{4\gamma + 2\rho})$, w.v.h.p. Noting that $t< T = n^{0.5-\eps}$, we have that Phase 4 is successfully completed. Overall, at some time $t_f = t'_s + t$, we have coalesced to a point $x(t_f) = y(t_f)$ with probability at least $1 - \Otilde(n^{-\eps})$.

\item \emph{Phase 5.} After the two processes have coalesced at the end of Phase 4, we allow them to progress under a complete coupling until they reach the same absorbing state.
\end{itemize}%
}%
\end{proof}

\InJournal{
Lemma~\ref{lem:coupling} almost completes the proof of Theorem~\ref{thm:rps}. We only need to note that the assumption on potential $U(x(0))$ made in Lemma~\ref{lem:coupling} can be applied to any point $x$ meeting the assumptions of Theorem~\ref{thm:rps}. The claim of the Theorem clearly holds for all points $x$ such that $-12 \geq U(x) \geq -\gamma \ln n$. Note, however, that the given upper bound on $U(x)$ can be dropped: for any point $x$ having potential too large, we can evolve it arbitrarily until a point with potential at most $-12$ is reached, and then apply the claim of the Theorem. Finally, we note that the required lower bound on $U(x)$ follows from the assumptions made in Theorem~\ref{thm:rps}: if $x_a > n^{-0.002}$ for all $a\in\{1,2,3\}$, then $U(x) \geq 2 \ln (n^{-0.002}) + \ln \frac{1}{3} > n^{-0.005}$ for sufficiently large $n$.
}
\InConference{
The proof of Theorem~\ref{thm:rps} is completed once we observe that the assumptions of Lemma~\ref{lem:coupling} are satisfied for any point $x$ satisfying the assumptions of Theorem~\ref{thm:rps}, either immediately (at time 0), or after a certain number of steps, once potential $U$ has been sufficiently reduced.
}
\InJournal{

We close the paper by remarking about the rate of convergence of RPS to an absorbing state. In the RPS protocol, as soon as one of the populations becomes empty, $x_a=0$, it is immediate to show that species $a+1$, without having any natural predators, will eliminate all representatives of species $a+1$, and the system will stabilize to the absorbing state: $x_{a+1} =1$, $x_a = x_{a+2} = 0$ within $\Oh (n \log n)$ steps. To bound the time until the first population is eliminated, we can perform a special case of the analysis from the proof of Theorem~\ref{thm:lvtime}, obtaining an overall bound of $\Otilde (n^2)$ on convergence time. This corresponds to the quantitative results obtained in the stochastic noise model presented in~\cite{DF12,RMF06}.
}

\section*{Acknowledgment}

The authors would like to thank Benjamin Doerr for helpful discussions. This research was partially funded by the EU IP FET Proactive project MULTIPLEX, by ANR project DISPLEXITY, and by NCN grant DEC-2011/02/A/ST6/00201.

\end{document}